\documentclass[journal,comsoc]{IEEEtran}

\usepackage{pifont,amssymb,url}
\usepackage{multicol,mathrsfs}
\usepackage[usenames,dvipsnames,svgnames,table]{xcolor}
\usepackage{bm,bbm,mathrsfs,amssymb,pifont,amssymb,pifont,amsfonts,amsmath,stmaryrd,color,amssymb,graphics,graphicx,epstopdf}
\usepackage{amscd,graphicx,array,dsfont,texdraw,tikz,enumerate,multicol,multirow,verbatim}
\usetikzlibrary{arrows,shapes,shadows,positioning,automata,patterns}
\usetikzlibrary{matrix,arrows,calc,trees,decorations.pathmorphing,decorations.markings, decorations.pathreplacing}
\usepackage[utf8x]{inputenc}
\usepackage[T1]{fontenc}
\usepackage{booktabs}
\usepackage{algorithm,algorithmicx,algpseudocode} 
\usepackage{bm,amsmath,amssymb,amsthm,graphicx,subfigure,epstopdf}
\usepackage[comma,numbers,square,sort&compress]{natbib}
\usepackage{epstopdf,epic,subfigure,epsfig,arydshln,enumerate}
\usepackage{fancyhdr,lastpage,color,dsfont}
\usepackage{tikz}
\usetikzlibrary{positioning}
\newtheorem{prop}{\it Property}
\newtheorem{lemma}{\bf \em{Lemma}}
\newtheorem{thm}{\bf \em{Theorem}}
\newtheorem{rem}{\bf \em{Remark}}
\newtheorem{problem}{\bf \em{Problem}}
\newtheorem{defi}{Definition}

\newcommand{\col}{\textnormal{col}}

\usepackage{textcomp}
\usepackage{arydshln} 
\def\BibTeX{{\rm B\kern-.05em{\sc i\kern-.025em b}\kern-.08em
    T\kern-.1667em\lower.7ex\hbox{E}\kern-.125emX}}
\allowdisplaybreaks

\begin{document}
\title{ Partial Excitation in Parameter Learning}
\author{Ganghui Cao, Shimin Wang, Martin Guay, Jinzhi Wang, Zhisheng Duan, Marios M. Polycarpou
\thanks{This work was supported by the National Natural Science Foundation of China under Grant T2121002 and Grant 62373352. (Corresponding author: Jinzhi Wang.)}
\thanks{Ganghui Cao, Jinzhi Wang, and Zhisheng Duan are with the State Key Laboratory for Turbulence and Complex Systems, Department of Mechanics and Engineering Science, College of Engineering, Peking University, Beijing 100871, China (e-mail: cgh@stu.pku.edu.cn, jinzhiw@pku.edu.cn, duanzs@pku.edu.cn). Shimin Wang is
with the School of Data Science, Lingnan University, Tuen Mun, Hong Kong (e-mail: smwang@ln.edu.hk). Martin Guay is with Queen's University, Kingston, ON K7L 3N6, Canada (e-mail: guaym@queensu.ca). Ganghui Cao and Marios M. Polycarpou are with the KIOS Research and Innovation Center of Excellence and the Department of Electrical and Computer Engineering, University of Cyprus, Nicosia 1678, Cyprus (e-mail: cao.ganghui@ucy.ac.cy, mpolycar@ucy.ac.cy).}
}
\maketitle

\begin{abstract}
This paper investigates parameter learning problems under {Partial Persistent Excitation (PPE)}. 
The PPE condition is a rank-deficient, and therefore, a more general evolution of the well-known {Persistent Excitation (PE)} condition.
Under the PPE condition, a proposed online algorithm is able to calculate the {PE} and non-{PE} subspaces, and finally gives an optimal parameter estimate in the sense of least squares. 
In particular, the learning error within the {PE} subspace exponentially converges to zero in the noise-free case. 
%
The PPE condition also provides a new perspective for solving distributed parameter learning problems, where the challenge is posed by local regressors that are often insufficiently excited.
To improve knowledge of the unknown parameters, a cooperative learning protocol is proposed for a group of estimators that collect measured information under complementary PPE condition.
This protocol allows each local estimator to operate locally in its {PE} subspace, and reach a consensus with neighbors in its non-{PE} subspace. 
As a result, the task of estimating unknown parameters can be achieved in a distributed way using cooperative local estimators. 
Application examples in system identification are given to demonstrate the effectiveness of the theoretical results developed in this paper.
\end{abstract}


\section{Introduction}
\label{sec:introduction}
\IEEEPARstart{P}{arameter} learning problems arise from system identification \cite{Ljung1983,pillonetto2022regularized}, adaptive control \cite{boyd1983parameter,Kreisselmeier1990}, adaptive filtering and prediction \cite{Goodwin2009}, nonlinear output regulation \cite{wang2024nonparametric,wang2025nonparametric}, active traffic management \cite{su2020neuro} and battery health management \cite{che2023opportunities,che2023battery}.
For example, parameter learning plays an important role in monitoring the health of Lithium-ion batteries as illustrated in \cite{qian2024fully,che2023battery}, where the estimation of temperature parameters significantly improves the accuracy of battery health monitoring.
In addition, the authors in \cite{su2020neuro} effectively integrated parameter learning and model-based design using neural networks and policy iteration for active traffic management.
The dynamical systems considered in parameter learning problems are often described by linear regression models \cite{Chen1991}, which express the measured output signals using regressor vectors, unknown parameters, and measurement noise.
The goal of parameter learning problems is to learn dynamic models from measured data \cite{pillonetto2022regularized,Farrell2006}.
In such context, persistent excitation (PE) condition plays a crucial role in ensuring accurate model learning and stable system performance \cite{narendra1987persistent,Ioannou2012}. 
%
{As the authors in \cite{anderson1985adaptive, hsu1987bursting} demonstrated, adaptive systems can exhibit bursting phenomena with or without $\sigma$-modification, in the absence of the PE condition.}
%
%
Nevertheless, it is well known that the PE condition is sufficient but not necessary for the convergence of parameter learning errors.
As a result, several efforts have been made to relax the PE condition. 

{There have been several direct variations of the PE
condition. 
A notable subclass of them is often referred to as the interval excitation (IE), which redefines the PE condition over a finite time interval rather than the infinite one. 
The IE condition arises from different contexts, such as adaptive control \cite{Kreisselmeier1990,adetola2008finite,adetola2010performance}, concurrent learning \cite{Chowdhary2013,Kamalapurkar2017}, and composite learning \cite{Pan2018}.
Recently, using an equivalent formulation of IE, the authors in \cite{LeiWang2024IJC} have given a definition of identifiability for linear regression models.
In addition to IE, some other direct generalizations of PE have been proposed in \cite{Praly2017,Barabanov2017}, which share the same features as the classical PE condition but with more elaborate characterizations. 
Specifically, the uniform width of the integration window and the uniform excitation level in the classical PE condition are allowed to vary.
Following a similar technical approach, a direct generalization of PE, referred to as weak persistent excitation, was proposed in \cite{Bruce2021}. 
Moreover, a class of recursive least squares estimators was studied in \cite{Bin2022}, where the proposed excitation condition offered some freedom to encompass and generalize the PE condition.
}

{Rather than relaxing the PE condition in a direct manner, a method referred to as dynamic regressor extension and mixing (DREM) was proposed in  \cite{Aranovskiy2016,Aranovskiy2017TAC}.} 
%
%
A key feature of the DREM method is the transformation of the regressor from its original vector form into a new scalar form, which yields interesting new convergence conditions for parameter estimation.
{These conditions have been proved to be no more restrictive, or even strictly weaker in some cases than the PE condition imposed on the original regression model \cite{Ortega2021TAC,Aranovskiy2023}.
}
%
%
The IE condition and even weaker excitation conditions were revisited and investigated to estimate the entire parameter vector within the DREM framework \cite{Wang2020,Aranovskiy2023}.
In addition, in stochastic regression models, the strong consistency of parameter estimation (i.e., the estimate converges to the true parameter with probability one) was also studied and established under some excitation conditions weaker than PE \cite{Lai1982,Chen1991}.

{The aforementioned studies were in pursuit of the weakest excitation conditions necessary to achieve full parameter estimation in adaptive systems.
In contrast, the authors in \cite{Bittanti1990,Marino2022,Tomei2023,Uzeda2023ARC,Uzeda2023SCL,Uzeda2024,Tomei2025} investigated the problem of partial parameter estimation. 
They focused on estimating parameters in a subspace under some excitation conditions insufficient to reconstruct the entire parameter vector.
In \cite{Marino2022}, a characterization of the lack of persistency of excitation is proposed, which can be computed online through an algorithm developed in \cite{Tomei2023}.
A novel subspace estimator is proposed in \cite{Uzeda2024}, which recovers a non-PE subspace defined based on the regressor's autocovariance matrix, by applying principal component analysis.
The common philosophy among these studies is to distinguish the identifiable part of unknown parameters from the non-identifiable one.
As a result, the parameter estimation methods therein are guaranteed to be robust in the absence of the PE condition.
The {\emph{Partial Persistent Excitation (PPE)}} condition formulated in this paper follows this philosophy, leading to new developments in parameter estimation.
}

The key motivation is that in distributed or large-scale network systems, local parameter estimators often have insufficiently exciting regressors and limited measurements.
This is caused by the insufficient richness of local inputs, as well as the limited capability of a single sensor, as revealed by various practical applications in \cite{wang2024distributed}.
%
Distributed parameter learning arises in a context where a group of sensor nodes individually collect local measurements to learn cooperatively a vector of unknown parameters.
{It is intriguing that cooperation at the group level can be achieved only through communication among neighboring nodes.}
%
The distributed parameter learning problem has been explored under various conditions and in different scenarios.
Earlier research, such as \cite{Chen2014}, studied the problem over undirected communication graphs.
The works in \cite{Matveev2022,Yan2023} investigated more general communication scenarios, on the premise that at least one of the sensor nodes collected sufficiently rich measurements for full parameter estimation. 
Moreover, the works in \cite{Javed2022,Garg2023} addressed the case that each sensor node collected insufficient measurements for full parameter estimation. 
It should be noted that the convergence of distributed parameter estimation was established in the absence of measurement noises in \cite{Javed2022,Garg2023}. 
Therefore, one key challenge that remains is how one can optimize the distributed parameter estimate in the presence of noise. In addition to the aforementioned works, there has been some important research contributions conducted within a probabilistic framework \cite{Cattivelli2010,Mateos2012,Zhang2012,Xie2018Auto,Xie2018SIAM,Yan2023}. 
Usually, the analysis procedures and obtained results therein relied on stringent assumptions about some statistical properties, such as moment conditions and white characters for the noise processes, independence and stationarity for the regressor processes, etc. 
Some recent results presented in \cite{Azzollini2023,Gan2024} were obtained under milder assumptions, at the cost of communicating more information than just local parameter estimates. 

{ In what follows, the contributions of this paper are presented.
In contrast to \cite{Marino2022,Tomei2023,Glushchenko2023,Uzeda2023ARC,Uzeda2023SCL,Uzeda2024,Tomei2025}, this paper makes the following progress:}
\begin{enumerate}
	\item\label{contri1} Under the PPE condition, the study proposes a parameter learning method that produces parameter estimates that are optimal in the sense of least squares. Specifically, unknown parameters are learned by minimizing a cost function of the estimation errors with a forgetting factor, which improves accuracy and alertness in learning parameters.
	\item\label{contri2} Based on the notion of PPE and the method introduced in \ref{contri1}), a distributed learning method is developed that provides a new perspective on the solution of the distributed parameter estimation problem with favorable features.
\end{enumerate}
The proposed parameter learning algorithm offers several notable advantages.
{
First, it allows the regressor to be partially exciting, in which case the standard least squares algorithm \cite{Ioannou2006} is neither robust\footnote{Those robust variations \cite{Ioannou2012}, though still referred to as ``least squares algorithms", are not optimal either.} nor optimal. 
Second, it produces an optimal parameter estimate with an exponential convergence rate. The parameter estimation update is developed in a recursive form, avoiding the need for the computation of a matrix inversion.
Third, it achieves robust parameter estimation through adaptively separating and taking advantage of the PE and non-PE parts of the regressor.
Moreover, it is used to develop a new methodology to tackle the problem of distributed parameter learning in contrast to existing studies found in the literature \cite{Chen2014,Javed2022,Garg2023,Matveev2022,Yan2023,zhang2025distributed}.
}
Some distinguishing features of this novel approach include:
\begin{enumerate}
	\item[1)] It integrates local optimizations into the distributed algorithm, which can be designed and implemented locally at each node, enabling a good scalability of sensor networks. The local optimizations can enhance the performance of cooperative parameter learning, leading to accurate parameter estimates by leveraging localized information collected by each sensor node.
	\item[2)] It allows the sensor nodes to communicate over a directed and unbalanced communication graph, which is a weak communication assumption for the distributed parameter estimation problem. Moreover, the distributed parameter learning is designed to guarantee an exponential rate of convergence on the overall sensor network.
	\item[3)] It applies to deterministic regression models without specific statistical assumptions. This makes the method less dependent on the statistical properties of the collected data, and, therefore, more widely applicable in practical situations.
\end{enumerate}

{%
The rest of this paper is organized as follows. In Section \ref{section2}, the problem formulation, the PE and PPE definitions, and the aims of the parameter learning approach are introduced. 
 The novel optimal parameter learning method is presented in Section \ref{section3}. 
 Based on the proposed method in Section \ref{section3}, a distributed parameter learning algorithm under complementary PPE condition is introduced in Section \ref{section4}.
 Applications in system identification with numerical examples are given in Section \ref{section5} to verify and illustrate the proposed design.
 Finally, conclusions are made in Section \ref{conclution}.}

\subsection*{Notation}
For a vector $x$ and a matrix $X$, $\left\| x \right\|$ and $\left\| X \right\|$ denote the Euclidean norm and the induced 2-norm, respectively. Let ${\rm {Im}}X$ denote the range or image of $X$, and ${\rm {Ker}}X$ denote the kernel or null space of $X$. Let $\lambda _{\min }(X)$ denote
the minimum eigenvalue of $X$, if $X$ is symmetric. 
For a complex number $\lambda$, denote  its real part by ${\rm Re}(\lambda)$.
For a set of matrices $\left\{X_i|\ i=1, 2, \dots, N\right\}$ and a set of their index $\mathcal{N}=\left\{1, 2, \cdots, N\right\}$, define ${\rm{diag}}{(X_1, \dots, X_N)}$ as the matrix formed by arranging the above matrices in a block diagonal fashion, and ${\rm{col}}{(X_1, \dots, X_N)}$ as a matrix formed by stacking them (i.e., 
$\setlength\arraycolsep{1.4pt}
{\left[ {\begin{array}{*{20}{c}}
			{X_1^{\top}}&{X_2^{\top}}& \dots &{X_N^{\top}}
	\end{array}} \right]^{\top}}$) if dimensions matched.
${\mathbf{1}_r}$ denotes a column vector of $1$’s of size $r$. $I$ and $0$ denote the identity matrix and zero matrix of appropriate dimensions, respectively.
A time-varying vector $x(t)$ is said to exponentially converge to zero at a decay rate no slower than $\rho$, if there exists a constant $\rho_x>0$ such that $\left\| x(t)  \right\| \le {\rho _x}{{\rm{e}}^{ - {\rho}t}}$. { Let ${h_{(i)}}$ denote the $i$th column of identity matrix $I_{n}$, for $i=1,\dots, n$.}

\section{Preliminaries} \label{section2}
\subsection{Problem Formulation} \label{PF}
Consider a continuous-time linear regression model
\begin{equation} \label{model}
	z(t) = {\phi }^{\top}(t)\theta  + \varepsilon (t),
\end{equation}
where $\phi \in \mathbb{R}^n$ is a smooth uniformly bounded vector referred to as the regressor, $\theta \in \mathbb{R}^n$ is a constant (or slowly varying) parameter to be estimated, $z \in \mathbb{R}$ is a continuous measurement, and $\varepsilon$ is a bounded measurement noise. A well-known assumption for the regressor is the persistent excitation (PE)  \cite{Ioannou2012} defined as follows.

\begin{defi}[\textbf{Persistent Excitation}]\label{definition1} The regressor $\phi(t)$ is said to display persistent excitation if there exist positive reals $T$, ${k_a}$, and ${k_b}$ such that
\begin{align*}
	{k_a} I_n \le \int_t^{t + T} {\phi (\tau ){\phi }^{\top}(\tau ){\rm{d}}\tau }  \le {k_b}{I_n},& & \forall t \ge 0.
\end{align*}
\end{defi}
In this paper, however, the parameter learning problem is studied under a variation of the PE concept
defined as follows.

{
\begin{defi}[\textbf{Partial Persistent Excitation}]\label{definition2} The regressor $\phi(t)$ is said to display {partial persistent excitation} with lack of persistency of order $q$ ($0\leq q\leq n$) if there exist positive reals $T$, $k_a$, $k_b$, and a real orthogonal projection matrix\footnote{A real matrix $\varPhi$ is referred to as an orthogonal projection matrix if $\varPhi^2=\varPhi=\varPhi^\top$.} $\varPhi$ of rank $n-q$ such that
\begin{align}\label{PPEcondition}
	k_a{\varPhi} \le \int_t^{t + T} {\phi (\tau ){\phi }^{\top}(\tau ){\rm{d}}\tau }  \le k_b{\varPhi},& & \forall t \ge 0.
\end{align}
\end{defi}
}

\begin{rem}\label{remark1} Definition~\ref{definition2} is mostly inspired by {the lack of persistency of excitation characterized in} \cite{Marino2022,Tomei2023}, which can also date back to \cite{Sadegh1990,Bittanti1990,Zhonghua1998}.
The PPE condition is weaker than the PE and coincides with the PE condition in the case of $q=0$. 
%
It can also be observed that the PPE condition always holds if the regressor is periodic. 
Taking the regressor $\phi (t) =\textnormal{col}({\sin t}, { - \sin t})$ as an example, it lacks persistency of order $1$, with 
\begin{equation*}
\underbrace{\left[ \begin{matrix}
			1&{ - 1}\\
			{ - 1}&1
	\end{matrix} \right]}_{{k_a\varPhi}}\le \int_t^{t + \pi} {\phi (\tau ){\phi }^{\top}(\tau ){\rm{d}}\tau }   \le \underbrace{\left[ \begin{matrix}
	2&{ - 2}\\
	{ - 2}&2
\end{matrix} \right]}_{{k_b\varPhi}}.
\end{equation*}
{The role of $k_b\varPhi$ in Definition~\ref{definition2} is to perfectly capture the rank-deficient case, since a persistently exciting regressor ${\phi }(t)$ naturally satisfies the first inequality in \eqref{PPEcondition} for some $k_a$.
In Definition~\ref{definition1}, the presence of ${k_b}{I_n}$ is mainly for technical purposes. It is used to rule out unbounded cases, so that parameter estimation algorithms can be carried out in practice.}
\end{rem}

Under the PPE condition given in Definition~\ref{definition2} with $q$, $T$, {$k_a$, $k_b$, and $\varPhi$} all unknown, we consider the following least squares problem: 
\begin{problem}
Minimize the cost function 
	\begin{align}\label{cost}
		J\left( {\vartheta (t)} \right)&=\frac{1}{2}\int_0^t {{{\rm{e}}^{ - \beta (t - \tau )}}{{\left( {z(\tau ) - {\vartheta }^{\top}(t)\phi (\tau )} \right)}^2}{\rm{d}}\tau }  \nonumber\\
		&\ \ \ + \frac{\alpha }{2}{{\rm{e}}^{ - \beta t}}{\big\| {\vartheta (t) - {{\hat \theta }_0}} \big\|^2}
	\end{align}
with respect to $\vartheta (t)$ at any given time $t$. 
\end{problem}
In the cost function, the positive real $\alpha$ reflects the degree of trust in the prior estimate $\vartheta (0) = {\hat \theta _0}$. The integral action penalizes all the past errors from $\tau = 0$ to $t$ with a forgetting factor $\beta>0$. Discounting the past data by $\beta$ helps keep the cost function alert to a slowly varying parameter. 

The {\bf{first objective}} is to design an online (recursive) algorithm to produce a parameter estimate $\hat \theta (t) \in \mathbb{R}^n$ such that:
\begin{enumerate}
	\item\label{amis1} in the absence of measurement noise, $\hat \theta (t)$ exponentially converges to $\theta$ in a subspace of rank $n-q$, referred to as the {PE} subspace.
	\item\label{amis2} in the presence of measurement noise, $\hat \theta (t)$ exponentially converges to ${\theta ^*(t)} = \mathop {\arg \min } \limits_{\vartheta(t)} J\left( {\vartheta}(t) \right)$, referred to as the least squares solution.
\end{enumerate}

Next, based on the designed online algorithm, consider solving the following distributed parameter estimation problem. Assume that there are $N$ measurements
\begin{align} \label{Nmeasure}
	{z_i}(t) &= \phi _i^{\top}(t)\theta + {\varepsilon _i}(t),& i \in {\cal N} = \left\{ {1, \cdots ,N} \right\},
\end{align}
with $N$ corresponding local estimators, none of which possessing local regressors ${\phi_i} \in \mathbb{R}^n$ that are persistently exciting. More precisely, each local regressor lacks persistency of excitation of order $q_i$, {\it i.e.},
{
\begin{align} \label{varPhiiab}
	{k_{ia}\varPhi_{i}} \le \int_t^{t + T} {\phi_i (\tau ){\phi_i^{\top}}(\tau ){\rm{d}}\tau }  \le {k_{ib}\varPhi_{i}},& & \forall t \ge 0,
\end{align} 
with $k_{ia}$ and $k_{ib}$ two positive reals, and $\varPhi_{i}$ a real orthogonal projection matrix of rank $n-q_i$.}
This implies that a local estimator can measure and estimate the parameter only in a subspace, {\it i.e.}, its {PE} subspace. As a result, the {\bf{second objective}} is to design a distributed cooperation strategy such that each local estimator can produce an estimate ${\hat \theta _i}(t) \in \mathbb{R}^n$ for the parameter in the whole space, {\it i.e.},
\begin{enumerate}
	\item[3)]\label{amis3} ${\hat \theta _i}(t)$ exponentially converges to $\theta$ for all $i\in {\cal N}$, in the absence of measurement noise, and,
	\item[4)]\label{amis4} ${\hat \theta _i}(t)$ exponentially converges to a neighborhood of $\theta$ for all $i\in {\cal N}$, in the presence of measurement noise.
\end{enumerate}
Here, distributed cooperation means that each local estimator communicates only with one or several neighbors, over a communication network modeled by a directed graph as defined in the following section.

\subsection{Directed Communication Graph}
A directed communication graph $\mathcal{G}=(\mathcal{N},\mathcal{E})$ is composed of a finite nonempty node set $\mathcal{N}=\left\{1,2,\cdots,N\right\}$, and an edge set $\mathcal{E}\subseteq \mathcal{N}\times \mathcal{N}$, in which the elements are ordered pairs of nodes. An edge originating from node $j$ and ending at node $i$ is denoted by $(j, i)\in \mathcal{E}$, which represents the direction of the message passing between the two nodes. 
The adjacency matrix of $\mathcal{G}$ is defined as $\mathcal{A}=[a_{ij}]\in \mathbb{R}^{N\times N}$, where $a_{ij}$ is a positive weight of the edge $(j,i)$ when $(j,i)\in \mathcal{E}$, otherwise $a_{ij}$ is zero. Assume that there are no self loops, i.e., $a_{ii}=0$, $\forall i \in \mathcal{N}$. The Laplacian matrix $\mathcal{L}=[l_{ij}]\in \mathbb{R}^{N\times N}$ of graph $\mathcal{G}$ is constructed by letting $l_{ii}=\sum\nolimits_{k = 1}^N {{a_{ik}}}$ and $l_{ij}=-a_{ij},\ \forall i,j \in \mathcal{N},\ i\neq j$.
A directed path from node $i$ to node $j$ is a sequence of edges $(i_{k-1},\ i_{k})\in \mathcal{E},\ k=1,2,\cdots,\bar k$, where $i_0=i,\ i_{\bar k}=j$. A directed graph $\mathcal{G}$ is said to be strongly connected if there exists at least one directed path from node $i$ to node $j$, $\forall i,j\in\mathcal{N},\ i\neq j$. A more comprehensive description of graph theory can be found in \cite{lewis2013cooperative,zhang2025distributed}.

\subsection{Supporting Lemmas}
\begin{lemma}\cite{Kim2020,Cao2023}\label{lemma1} For a strongly connected directed graph $\mathcal{G}(\cal N)$, there exists a vector $\xi  = {\rm{col}}{\left( \xi_1,\dots, \xi _N\right)} \in \mathbb{R}^N$ such that ${\xi ^{\top}}{\cal L} = 0$, $\mathbf{1}_N^{\top}\xi  = 1$, and ${\varXi _0} = {\rm{diag}}{\left( \xi_1,\dots, \xi _N\right)} > 0$. 
In addition, given matrices $X_i\in{\mathbb{R}^{n\times q_i}}$ satisfying $X_i^\top X_i=I_{q_i}$, ${\forall i \in {\cal N}}$, there is $${\rm{diag}}(X_1,\dots, X_N)^\top (\hat{\mathcal{L}}\otimes I_n){\rm{diag}}(X_1,\dots, X_N)>0$$ with $\hat{\mathcal{L}}={\varXi_0}{\cal L} + {{\cal L}^{\top}}{\varXi_0}$, if and only if $\cap_{i=1}^{N} {\rm{Im}}X_i=\left\lbrace 0 \right\rbrace $.
\end{lemma}
\begin{lemma}\cite{Boyd1994}\label{lemma2} The matrix $\varUpsilon^*$ is Hurwitz, or equivalently, all trajectories of the differential equation $\dot{x}(t)=\varUpsilon^* x(t)$ converge to zero), if and only if there exists a positive definite matrix $\varXi$, such that $\varUpsilon^{*\top}\varXi+\varXi \varUpsilon^*<0$.
\end{lemma}
\begin{lemma}\label{lemma3} Consider the following linear time-varying dynamical system: 
\begin{equation}\label{lem_sys1}
	\dot x(t) = \varUpsilon (t)x(t) + u(t),
\end{equation}
where $x$ is the state vector, $u$ is the input vector, and $\varUpsilon$ is a square matrix of appropriate size. Suppose there exist positive reals ${\rho_a}$, ${\rho_b}$, ${\rho_c}$, and ${\rho_d}$ such that 
\begin{equation*}
	\left\| {\varUpsilon (t) - {\varUpsilon ^*}} \right\| \le {\rho _a}{{\rm{e}}^{ - {\rho _b}t}}\ \text{and}\ \left\| {u(t) - {u^*}(t)} \right\| \le {\rho _c}{{\rm{e}}^{ - {\rho _d}t}},
\end{equation*}
where ${\varUpsilon^*}$ is a stable matrix with all eigenvalues lying in the half-plane ${\rm{Re}}(s) \le  - \upsilon$, and ${u^*}(t)$ is a bounded time-varying signal. Then, for any ${\rho _f}$ satisfying $0 < {\rho _f} < \min \left\{ {\upsilon ,{\rho _b},{\rho _d}} \right\}$, there exists a positive real ${\rho_e}$ such that 
	$$\left\| {x(t) - {x^*}(t)} \right\| \le {\rho _e}{{\rm{e}}^{ - {\rho _f}t}},$$
where $${x^*}(t) = \int_0^t {{{\rm{e}}^{{\varUpsilon ^*}(t - \tau )}}{u^*}(\tau ){\rm{d}}\tau } .$$ In particular, if ${u^*}(t)$ exponentially converges to zero at a decay rate no slower than ${\rho _g}$, then $x(t)$ exponentially converges to zero at a decay rate no slower than any ${\rho _h} < \min \left\{ {\upsilon ,{\rho _b},{\rho _d},{\rho _g}} \right\}$.
\end{lemma}
See Section \ref{apendix-lemma3} for the Proof of Lemma~\ref{lemma3}. 

\section{Parameter Learning under Partial Persistent Excitation}\label{section3} \label{SLSM}
This section focuses on achieving aims \ref{amis1})--\ref{amis2}) formulated in Section \ref{PF}. Since the order of deficiency of excitation may be unknown {\it a priori}, an online algorithm is designed to identify adaptively which of the orders are lacking. Based on that, an online algorithm for Parameter Learning is developed.

\subsection{Define the {PE and non-PE} Subspaces}
{
The orthogonal projection matrix $\varPhi$ in Definition~\ref{definition2} has the singular value decomposition
\begin{align*} \label{singdecomPhi}
	{\varPhi} = \left[ {\begin{array}{*{20}{c}}
			{{N_{d}}}&{{N_{u}}}
	\end{array}} \right]\left[ {\begin{array}{*{20}{c}}
			{{I_{n-q}}}&{}\\
			{}&{{0_{q \times q}}}
	\end{array}} \right]\left[ {\begin{array}{*{20}{c}}
			{N_{d}^{\top}}\\
			{N_{u}^{\top}}
	\end{array}} \right],
\end{align*}
where $\left[ {\begin{array}{*{20}{c}}
		{{N_{d}}}&{{N_{u}}}
\end{array}} \right]$ is an orthogonal matrix.
According to this decomposition, by pre- and post-multiplication with $N_{d}^{\top}$ and ${N_{d}}$ respectively, the first inequality in \eqref{PPEcondition} becomes
\begin{equation}\label{PPE1}
	{k_a I_{n-q}} \le N_{d}^{\top}\left(\int_t^{t + T} {\phi (\tau ){\phi }^{\top}(\tau ){\rm{d}}\tau } \right){N_{d}}.
\end{equation}
By pre- and post-multiplication with $N_{u}^{\top}$ and ${N_{u}}$ respectively, the second inequality in \eqref{PPEcondition} gives
\begin{equation}\label{PPE3}
	N_{u}^{\top}\left(\int_t^{t + T} {\phi (\tau ){\phi }^{\top}(\tau ){\rm{d}}\tau }\right) {N_{u}} = {0_{q \times q}}.
\end{equation}
Define ${\rm{Im}}{N_d}$ as the {PE} subspace. It will be seen later that the parameter can only be identified in this subspace. Correspondingly, define ${\rm{Im}}{N_u}$ as the non-{PE} subspace.
}

\subsection{Calculate the {PE} Subspace} \label{CIS}
The purpose of this subsection is to estimate ${N_d}N_d^{\top}$, rather than ${N_d}$ directly. There are two benefits to doing so: 
\begin{itemize}
	\item The matrix size of ${N_d}N_d^{\top}$ is $n \times n$, which is fixed and independent of the unknown column numbers of ${N_d}$. {(Note that the order of lack of persistency is unknown.)}
	\item {The matrix ${N_d}N_d^{\top} = \varPhi$ is unique and independent of the specific choices of ${k_a}$, ${k_b}$, and ${N_d}$, leading to a one-to-one correspondence between ${N_d}N_d^{\top}$ and the {PE} subspace.
    }
\end{itemize}
The following algorithm is designed to estimate ${N_d}N_d^{\top}$:
\begin{subequations}\label{QP}
	\begin{align}
		\dot Q(t)  &=  - \beta Q(t) + \! \phi (t){\phi }^{\top}(t),  \label{Q}\\
		\dot P(t)  &=  - \gamma P(t) + \! \gamma I -\! {\gamma ^2}\!\!\! \int_0^t \!\! {{{\rm{e}}^{ - \gamma (t - \tau )}}{{\bar N}_u}(\tau )\bar N_u^{\top}(\tau ){\rm{d}}\tau },  \label{P}\\
		{{\bar N}_u}(\tau )  &= {{\hat N}_u}(k\delta ),\ \ \ \ \ \ \ \ \ \ \ \ \ \ \ \ \ \ \ k\delta  \le \tau  < \left( {k + 1} \right)\delta, \label{N}
	\end{align}
\end{subequations}
where $Q(0)=P(0)=0_{n\times n}$, $\beta$ appears in the cost function \eqref{cost}, $\gamma$ and $\delta $ are arbitrarily chosen finite positive reals, ${\hat N_u}(k\delta)$ is a matrix formed by an orthonormal basis of ${\rm{Ker}}Q(k\delta)$\footnote{In other words, the column vectors of ${\hat N_u}$ are the right singular vectors of $Q$ corresponding to zero singular values, and so can be obtained from singular value decomposition.}, and $k$ is a nonnegative integer used to locate the interval in which $\tau $ resides.

It should be noted that the eigenspaces of a continuously varying matrix are not necessarily continuous \cite{Kato1995}. Therefore, ${{\hat N}_u}\hat N_u^{\top}$ may not be a continuous function of time, even though $Q(t)$ is continuous in time. Given this unfavourable fact, the role of \eqref{P} and \eqref{N} is to generate a continuously differentiable estimate for ${N_d}N_d^{\top}$ from the information of a possibly discontinuous matrix signal ${{\hat N}_u}\hat N_u^{\top}$. The differentiability of $P$ paves the way for the subsequent algorithm design.

It should also be noted that a possibly discontinuous matrix ${{\hat N}_u}\hat N_u^{\top}$, even if bounded, may not be integrable (for example, in the case of having an oscillating discontinuity). Applying \eqref{N} can obtain an integrable matrix ${{\bar N}_u}\bar N_u^{\top}$, which guarantees a well-defined integral in \eqref{P}. Moreover, it reduces the computational load, in the sense that ${\hat N_u}$ is only computed at a frequency of $\delta$.

\begin{thm}\label{Theorem1} If regressor $\phi(t)$ displays PPE with lack of persistency of order $q$, then matrix $P(t)$ given by algorithm \eqref{QP} is continuously differentiable, satisfying $0 \le {P(t)} \le I$, and there exist two positive reals ${\rho_a}$ and ${\rho_b}$ such that $$\left\| {P(t) - {N_d}N_d^{\top}} \right\| \le {\rho _a}{{\rm{e}}^{ - {\rho _b}t}}.$$ Moreover, the decay rate ${\rho_b}$ can be made arbitrarily fast by increasing $\gamma$.
\end{thm}
{ See Section~\ref{proofthe1} for the proof of Theorem~\ref{Theorem1}.}

\subsection{Parameter Learning Algorithm} \label{CTPE}
The parameter learning is made possible by the continuously differentiable estimate for ${N_d}{N_d^{\top}}$ given in the previous subsection with the following algorithm to estimate $\theta$: 
\begin{subequations}\label{hat_theta}
	\begin{align}
	{{\dot{\hat\theta}}_d} &= - \varOmega \left( {R{{\hat \theta }_d} - zP\phi  - \dot P\varphi } \right), \label{thetad}\\
	{\hat \theta_u} &= \left( {I - P} \right){\hat \theta _0}, \label{thetau} \\
	\hat \theta \ &= {\hat \theta _d} + {\hat \theta _u}, \label{theta}
	\end{align}
\end{subequations}
where ${{\hat \theta }_d}(0) = 0$, ${\hat \theta }_0$ is the prior estimate already defined in the cost function \eqref{cost}, and $\varphi$, $\varOmega$ and $R$ are generated by
\begin{align}
	\dot \varphi =& - \beta \varphi  + z\phi, &\varphi (0) = \alpha {{\hat \theta }_0}, \tag{\ref{hat_theta}{d}} \label{varphi}\\
	\dot \varOmega =&\ \beta \varOmega - \varOmega R\varOmega, & \varOmega (0) = {\kappa ^{ - 1}}I, \tag{\ref{hat_theta}{e}} \label{omega}\\
			R =&\ P\phi (t){\phi }^{\top}(t)P + \kappa \beta \left( {I - P} \right)&\nonumber\\
		& + \dot PQP + PQ\dot P + \big( {\alpha {{\rm{e}}^{ - \beta t}} - \kappa } \big)\dot P, &
\tag{\ref{hat_theta}{f}} \label{R}
\end{align}
with $\kappa$ an arbitrarily chosen finite positive real, and $Q$, $P$ and $\dot P$ given in \eqref{QP}.

\begin{thm}\label{Theorem2} If regressor $\phi$ displays PPE with lack of persistency of order $q$, then the algorithm given by \eqref{hat_theta} guarantees that there exist two positive reals ${\rho_a}$ and ${\rho_b}$ such that $$\big\| {\hat \theta (t) - {\theta ^*}(t)} \big\| \le {\rho _a}{{\rm{e}}^{ - {\rho _b}t}},$$ where ${\theta^*}(t)$ is the least squares solution that minimizes the cost function $J$, and the decay rate $\rho_b$ can be made arbitrarily fast by increasing $\gamma$. In particular, there exists a positive real $\rho_c$ such that $$\big\| N_d^{\top}\big( {\hat{\theta} (t)  - \theta } \big) \big\| \le {\rho_c}{{\rm{e}}^{ - ({\beta}/{2})t}}$$ in the noise-free case $\varepsilon (t) \equiv 0$. 
\end{thm}
{ See Section~\ref{proofthe2} for the proof of Theorem~\ref{Theorem2}.}

\begin{rem}\label{remark2}
Combining \eqref{model} and \eqref{Nd_theta}, with the fact that zero $\varepsilon$ leads to \eqref{til_the_d_con}, it can be easily concluded that bounded $\varepsilon$ leads to bounded $N_{d}^{\top}\theta^*(t)$ for all $t\geq 0$. It then follows from \eqref{nec_con2} that $\theta^*(t)$ is also bounded for all $t\geq 0$.
\end{rem}

\begin{rem}\label{remark3} Although no prior knowledge of the {PE} and non-{PE} subspaces, algorithm \eqref{hat_theta} can adaptively update the estimate in the former subspace while leaving the estimates unchanged in the latter subspace. As a result, the obtained parameter estimate is, not only robust to noises but also optimal with exponential convergence in the sense of least squares. {In contrast, the standard and modified recursive least squares algorithms \cite{Ioannou2006,Ioannou2012} are not guaranteed to give a least squares estimate (even in the subspace where PE is satisfied), because their recursive forms are derived based on PE condition in the whole space.}
\end{rem}
\section{Distributed Parameter Learning Under Complementary Partial Excitation Condition}\label{section4}
The purpose of this section is to achieve aims \ref{amis3})--\ref{amis4}) formulated in Section \ref{PF}. 
In this section, a distributed parameter estimation algorithm is given first, followed by error dynamics analysis, and then the main results about the convergence of the algorithm are presented.

The distributed parameter learning algorithm is designed based on Section \ref{SLSM}. In distributed situations, the duplication of the algorithms \eqref{QP}, \eqref{thetad}, and \eqref{varphi}--\eqref{R} at each node yields
\begin{subequations}\label{dis_est}
	\begin{align}
		{{\dot{\hat\theta}}_{id}} = - \varOmega_i \left( {R_i{{\hat \theta}_{id}} - z_iP_i\phi_i(t)  - \dot P_i\varphi_i } \right). 
    \end{align}
\end{subequations}
The distributed learning through communication among neighbors is achieved by the parameter update
\begin{equation}
	{\dot {\hat \theta}_{iu}} =  - {\eta _{id}}{P_i}{\hat \theta _{iu}} - {\eta _{iu}}\left( {I - {P_i}} \right)\sum\nolimits_{j = 1}^N {{a_{ij}}\left( {{{\hat \theta }_i} - {{\hat \theta }_j}} \right)}, \tag{\ref{dis_est}{b}} \label{hat_the_iu}
\end{equation}
where the initial condition ${\hat \theta}_{iu}(0)$ is chosen as $\hat \theta_{i0}$, the prior estimate for $\theta$, $\eta _{id}$ and $\eta _{iu}$ are arbitrarily chosen finite positive reals number, and ${{\hat \theta }_i}$ is the parameter estimation computed as
\begin{equation}
	{\hat \theta _i} = {P_i}{\hat \theta _{id}} + \left( {I - {P_i}} \right){\hat \theta _{iu}}. \tag{\ref{dis_est}{c}} \label{hat_the_i}
\end{equation}

{
\begin{rem}\label{addremark} 
In \cite{Uzeda2024}, a $\mu$-modification is applied in the non-PE subspace, which provides robust adaptation without sacrificing the estimation error regulation. In \eqref{hat_the_iu}, a neighboring-feedback consensus strategy is brought into the local non-PE subspaces, which helps achieving full parameter estimation.
\end{rem}
}

The behavior of ${\hat \theta _{id}}$ has already been studied in Section \ref{CTPE}: According to Step 4 in \ref{proofthe2}, there exist $\rho_a,\ \rho_b>0$ such that
\begin{align} \label{hatthetaid-star}
	\left\| {\hat \theta _{id}}(t) - {N_{id}}{N_{id}^{\top}}\theta _i^*(t) \right\| \le {\rho _a}{{\rm{e}}^{ - {\rho _b}t}},
\end{align}
where the column vectors of $N_{id}$ form an orthonormal basis for the local {PE} subspace, and $\theta_i^*$ is defined as
\begin{align*}
	{\theta_i^*(t)} = \mathop {\arg \min } \nolimits_{\vartheta_i(t)} J_i\left( {\vartheta_i(t)} \right),\end{align*}
i.e., the least squares solution that minimizes the cost function 
\begin{align*}
		J_i\left( {\vartheta_i (t)} \right) =&\  \frac{1}{2}\int_0^t {{{\rm{e}}^{ - \beta_i (t - \tau )}}{{\left( {z_i(\tau ) - {\vartheta_i^{\top}}(t) \phi_i (\tau )} \right)}^2}{\rm{d}}\tau }  \\
		& + \frac{\alpha_i}{2}{{\rm{e}}^{ - \beta_i t}}{\big\| {\vartheta_i (t) - {{\hat \theta }_{i0}}} \big\|^2},
\end{align*}
with $\alpha_i$ the degree of trust in the prior estimate $\hat \theta_{i0}$. In addition, according to Step 5 in \ref{proofthe2}, there exists $\rho_c>0$ such that 
\begin{align} \label{Nthestar-Nthe_i}
	\left\| N_{id}^{\top}\theta _i^*(t) - N_{id}^{\top}\theta \right\| \le {\rho _c}{{\rm{e}}^{ - ({\beta_i}/{2})t}},
\end{align}
in the absence of measurement noises.

To assess the behavior of ${\hat \theta _{iu}}$ and ${\hat \theta _i}$, the following estimation error vectors are defined:
\begin{align} \label{til_def}
	{{\tilde \theta }_{id}} = {{\hat \theta }_{id}} - \theta,& & {{\tilde \theta }_{iu}} = {{\hat \theta }_{iu}} - \theta,& & {{\tilde \theta }_i} = {{\hat \theta }_i} - \theta. 
\end{align}
Then, from \eqref{hat_the_i},  we have the following equation 
\begin{equation} \label{til_rel}
	\begin{split}
		{{\tilde \theta }_i} &= {P_i}\left( {{{\hat \theta }_{id}} - \theta } \right) + \left( {I - {P_i}} \right)\left( {{{\hat \theta }_{iu}} - \theta } \right) \\
		&= {P_i}{{\tilde \theta }_{id}} + \left( {I - {P_i}} \right){{\tilde \theta }_{iu}}.
	\end{split}
\end{equation}
Utilizing \eqref{til_def} and \eqref{til_rel}, the dynamics of ${\tilde \theta _{iu}}$ is such that
\begin{equation} \label{dyn_til_the}
	\begin{aligned}
		{{\dot {\tilde \theta}}_{iu}} =&  - {\eta _{id}}{P_i}{{\hat \theta }_{iu}} - {\eta _{iu}}\left( {I - {P_i}} \right)\sum\nolimits_{j = 1}^N {{a_{ij}}\left( {{{\tilde \theta }_i} - {{\tilde \theta }_j}} \right)} \\
		=&  - {\eta _{id}}{P_i}{{\hat \theta }_{iu}} - {\eta _{iu}}\left( {I - {P_i}} \right)\sum\nolimits_{j = 1}^N {{l_{ij}}{{\tilde \theta }_j}} \\
		=& - {\eta _{id}}{P_i}{{\hat \theta }_{iu}} - {\eta _{iu}}\left( {I - {P_i}} \right)\sum\nolimits_{j = 1}^N {{l_{ij}}{P_j}{{\tilde \theta }_{jd}}} \\
		& \hspace{-0.25in}- {\eta _{iu}}\! \left( {I \!-\! {P_i}} \right)\!\!\sum\limits_{j = 1}^N \!{{l_{ij}}\!\!\left( {I \!-\! {P_j}} \right)\!\!{\left( {{N_{ju}}{N_{ju}}^{\top} \!+\! {N_{jd}}{N_{jd}}^{\top}} \right)}{{\tilde \theta }_{ju}}}.
	\end{aligned}
\end{equation}
For each agent, the dynamics \eqref{dyn_til_the} are pre-multipled by constant matrices ${N_{iu}^{\top}}$, whose row vectors form an orthonormal basis for the local non-{PE} subspace. Then, by considering all nodes, the overall error dynamics system can be written in the following compact form:
	\begin{align}\label{ove_dyn}
		N_U^{\top}{{\dot {\tilde \theta} }_U} =&  - {H_U}{N_U^{\top}}{P_U}\left( {{\cal L} \otimes I_n} \right){P_U}{N_U}{N_U^{\top}}{{\tilde \theta }_U}\nonumber\\
		& - {H_U}{N_U^{\top}}{P_U}\left( {{\cal L} \otimes I_n} \right){P_U}{N_D}{N_D^{\top}}{{\tilde \theta }_U} \nonumber\\
		& - {H_U}\! {N_U^{\top}}{P_U}\! \left( {{\cal L} \otimes I_n} \right)\!{P_D}{{\tilde \theta }_D} \!-\! {H_D} {N_U^{\top}}{P_D}{{\hat \theta }_U},
	\end{align}
where ${N_U} = {\rm{diag}}\big( N_{1u}, \dots, N_{Nu}\big)$, ${{\tilde \theta }_U} = {\rm{col}}{\big( {\tilde \theta }_{1u}, \dots, {\tilde \theta }_{Nu}\big)}$, ${N_D} = {\rm{diag}}\big( N_{1d}, \dots, N_{Nd}\big)$, ${{\hat \theta }_U} = {\rm{col}}{\big( {\hat \theta }_{1u}, \dots, {\hat \theta }_{Nu}\big)}$, ${{\tilde \theta }_D} = {\rm{col}}{\big( {\tilde \theta }_{1d}, \dots,{\tilde \theta }_{Nd}\big)}$, ${H_U} = {\rm{diag}}{\big( {\eta _{1u}}{I_{{q_1}}}, \dots,{\eta _{Nu}}{I_{{q_N}}}\big)}$, ${H_D} = {\rm{diag}}{\big( {{\eta _{1d}}{I_{{q_1}}}},\dots, {{\eta _{Nd}}{I_{{q_N}}}}\big)}$, ${P_D} = {\rm{diag}}{\left( {{P_1}},\dots, {{P_N}}\right)}$, and ${P_U} = {\rm{diag}}{\left( {I_n - {P_1}},\dots, {I_n - {P_N}}\right)}$.

On the right-hand side of \eqref{ove_dyn}, the first term is the autonomous part, while the second, third and fourth terms all contribute to the nonautonomous part. It will be shown shortly that, under a complementary PPE condition, the following properties hold:

\begin{prop}\label{property1}
The coefficient matrix of the autonomous part exponentially converges to a stable matrix.
\end{prop}

\begin{prop}\label{property2} The nonautonomous part exponentially converges to zero in the absence of measurement noises.\end{prop}

\begin{prop}\label{property3} The nonautonomous part exponentially converges to a bounded set containing the origin in the presence of measurement noise.\end{prop}

The convergence of algorithm \eqref{dis_est} can be characterized with the help of the following reference system:
\begin{subequations}\label{refsys_cha}
\begin{align} 
	\tilde \theta _D^*(t) &= {{N_D^{\top}}}{\rm{col}}{\left( {\theta_1^*(t) - \theta }, \dots, {\theta_N^*(t) - \theta }\right)}\\
	{\dot {\tilde \theta}^*_U }(t) &= - {H_U}{N_U^{\top}}\left( {{\cal L} \otimes I} \right)\left({N_U}\tilde \theta _U^*(t)+{N_D}\tilde \theta _D^*(t)\right) \label{reference_solutionU}\\
	{\tilde \theta _I}^*(t) &= {N_D}\tilde \theta _D^*(t) + {N_U}\tilde \theta _U^*(t),
\end{align}
\end{subequations}
where ${\tilde \theta}^*_U(0)=0$, and $\theta _i^*$ is the optimal parameter estimate in the sense of minimizing $J_i$. In fact, the solution to \eqref{reference_solutionU} is 
\[\tilde \theta _U^*(t) =  - \!\! \int_0^t \!\!{{{\rm{e}}^{ - {H_U}{N_U^\top} \left( {{\cal L} \otimes I} \right){N_U}(t - \tau )}}{H_U}{N_U^\top} \left( {{\cal L} \otimes I} \right){N_D}\tilde \theta _D^*(\tau ){\rm{d}}\tau }. \]

\begin{thm}\label{Theorem3} Suppose the regressor at the $i$th node ${\phi _i}$ lacks persistency of order $q_i$, the complementary PPE condition {$\sum\nolimits_{i = 1}^N {{\varPhi_{i}}}  > 0$} is satisfied, and the communication graph is strongly connected. Then the algorithm \eqref{dis_est} guarantees that there exist two positive reals $\rho_a$ and $\rho_b$ such that
\begin{align*}
	\left\| {\tilde \theta _I}(t) -{\tilde \theta _I}^*(t) \right\| \le {\rho _a}{{\rm{e}}^{ - {\rho _b}t}},
\end{align*}
where ${\tilde \theta _I} = {\rm{col}}{( {{{\tilde \theta }_1}},\dots, {{{\tilde \theta }_N}})}$ is the overall parameter estimation error vector, ${\tilde \theta _I}^*$ is the trajectory of system \eqref{refsys_cha}, and $\rho_b$ can be made arbitrarily large by increasing ${\gamma _i}$ and ${\eta _{iu}}$. 
In particular, for any $\rho_d < \mathop {\min }\limits_{i\in \mathcal{N}} \left\{ {{\beta _i}/2} \right\}$, there exists a positive real $\rho_c$ such that $\big\| {\tilde \theta _I}(t) \big\| \le {\rho_c}{{\rm{e}}^{ - \rho_d t}}$ in the noise-free case $\varepsilon_i (t) \equiv 0$.
\end{thm}
{ See Section~\ref{proofthe3} for the proof of Theorem~\ref{Theorem3}.}
\begin{rem}\label{remark4} 
Similar to Remark~\ref{remark2}, zero $\varepsilon _i$ leads to \eqref{Nthestar-Nthe_i}, and bounded $\varepsilon _i$ leads to the bounded function $N_{id}^{\top}\theta _i^*$.
Then it follows from \eqref{Nthestar-Nthe_i} and \eqref{Lamdacexp} that Properties~\ref{property2} and \ref{property3} hold true. 
In addition, since $- {H_U}{N_U^{\top}}\left( {{\cal L} \otimes I} \right){N_U}$ is Hurwitz, it is guaranteed that the trajectory ${\tilde \theta _I}^*$ of system \eqref{refsys_cha} is bounded if $\varepsilon _i$ is bounded, and converges to zero if $\varepsilon _i$ is zero.
\end{rem}
\section{Applications in System Identification}\label{section5}

This section provides two simulation examples of the proposed algorithms to demonstrate their possible applications in system identification.

\subsection{Application 1: Identification for Linear Systems}
Consider the identification problem for a linear time-invariant dynamical system
	\begin{align}\label{SISO}
		\dot x &= Fx + bu,\quad
		&y = h_{(1)}^{\top}x, 
	\end{align}
where $x\in \mathbb{R}^{n_F}$, $u\in \mathbb{R}$, and $y\in \mathbb{R}$ are the state, input, and output respectively, with unknown system parameters $F \in \mathbb{R}^{n_F \times n_F}$, $b,\ {h_{(1)}} \in \mathbb{R}^{n_F}$. The objective is to estimate the unknown parameters from the input and output of the system. If $\big( {F,h_{(1)}^{\top}} \big)$ is observable, it entails no loss of generality to suppose that
\begin{equation} \label{Fh}
	\setlength\arraycolsep{3 pt}
	F = \left[\begin{array}{*{20}{c:c}}
    f\ & \begin{array}{*{20}{c}}
    I_{{n_F} - 1}\\
    \hdashline
    {0_{1 \times ({n_F} - 1)}}
    \end{array}\end{array}
    \right]\ \ \textnormal{and}\ \ {h_{(1)}} = \left[ {\begin{array}{*{20}{c}}
	1\\
    \hdashline
	{{0_{({n_F} - 1) \times 1}}}
\end{array}} \right],
\end{equation} 
    with $f = \textnormal{col}({f_1}, \dots,f_{n_F})$ and $b =\textnormal{col}({b_1}, \dots, b_{n_F})$.
The state space representation \eqref{SISO} with \eqref{Fh} is referred to as the observable canonical form \cite{Bay1999}, which is equivalent to any other state space representation. Under this form, only $f$ and $b$ are unknown parameters that need to be estimated. Based on this form, one can finally arrive at (see Section \ref{ARS-repre} for details) the following algebraic representation of system \eqref{SISO}:
\begin{equation}\label{AR}
	y = h_{(1)}^{\top}{{\rm{e}}^{Wt}}x(0) + h_{(1)}^{\top}{\varPi _y}\left( {f - w} \right) + h_{(1)}^{\top}{\varPi _u}b, 
\end{equation} 
where $w = \textnormal{col}({w_1}, \dots ,{w_{n_F}})$ is a vector designed such that
\begin{equation} \label{W}
	\setlength\arraycolsep{3 pt}
	W = \left[\begin{array}{*{20}{c:c}}
    w\ & \begin{array}{*{20}{c}}
    I_{{n_F} - 1}\\
    \hdashline
    {0_{1 \times ({n_F} - 1)}}
    \end{array}\end{array}
    \right]
\end{equation}
is a stable matrix, and $r_u$ and $r_y$ generated by
\begin{subequations}
	\begin{align}
		{{\dot r}_u} &= {W^{\top}}{r_u} + {h_{(1)}}u,& {r_u}(0) = 0, \label{r_u}\\
		{{\dot r}_y} &= {W^{\top}}{r_y} + {h_{(1)}}y,& {r_y}(0) = 0,  \label{r_y}
	\end{align}
\end{subequations}
are both bounded signals. The matrices ${\varPi _u}$ and ${\varPi _y}$ in \eqref{AR} are written as
\begin{subequations}
\begin{align}
	{\varPi _u} &=H_W^{-1} \textnormal{\col}({r_u^{\top}}, {r_u^{\top}W},\dots, {r_u^{\top}{W^{{n_F} - 1}}}),\label{Pi_u_mat} \\
	{\varPi _y} &=H_W^{-1}\textnormal{\col}({r_y^{\top}},{r_y^{\top}W},\dots, {r_y^{\top}{W^{{n_F} - 1}}}), \label{Pi_y_mat}
\end{align}
\end{subequations}
where $H_W=\textnormal{col}({h_{(1)}^{\top}}, {h_{(1)}^{\top}} W,\dots, h_{(1)}^{\top}{W^{{n_F} - 1}})$.
The algebraic representation \eqref{AR} coincides with the regression model \eqref{model}, i.e.,
\begin{equation*}
	\underbrace{y + h_{(1)}^{\top}{\varPi _y}w}_z =\underbrace{h_{(1)}^{\top}\left[ {\begin{array}{*{20}{c}}
			{{\varPi _u}}&{{\varPi _y}}
	\end{array}} \right]}_{\phi^\top} \underbrace{\left[ {\begin{array}{*{20}{c}}
	b\\
	f
\end{array}} \right]}_{\theta} + \underbrace{h_{(1)}^{\top}{{\rm{e}}^{Wt}}x(0)}_{\varepsilon}. 
\end{equation*} 

\subsubsection*{\textbf{Numerical Example of Application 1}}
Let $n_F=3$, $b = \textnormal{col}(1,{ - 5},9)$, $f = \textnormal{col}({ - 2.5},{ - 11},{ - 5})$, and choose $w =\textnormal{col}({ - 4}, { - 9.25}, { - 6.25})$, ${{\hat \theta }_0} = {\mathbf{1}_{6 \times 1}}$, $\alpha  =1$, $ \beta  =1$, $\gamma  =1$, $\delta  =1$, and $\kappa  = 1$. {Consider the following exploration inputs 
$$u = 10\sum\nolimits_{j = 1}^k {\sin \left( {(2j - 1)t + 2j} \right)}, \ k = 1,2,3$$ 
for system \eqref{SISO}. The simulation results shown in Figs.~\ref{one_fre_noisefree} and \ref{three_fre_noisefree} are obtained in the absence of noise, whereas the results shown in Figs.~\ref{one_fre}, \ref{three_fre}, \ref{two_fre_con1}, and \ref{two_fre_con2} are obtained by adding white Gaussian noise with a standard deviation of 1 to $y$. In Figs.~\ref{two_fre_con1} and \ref{two_fre_con2}, the proposed algorithm \eqref{hat_theta} is compared with the following standard recursive least squares algorithm \cite{Ioannou2006}: 
\begin{align*}
	\dot {\breve \theta} &=  - \breve\Omega \phi ({\phi ^ \top }\breve \theta  - z),\ \breve \theta (0) = {{\hat \theta }_0}\nonumber\\ 
	\dot {\breve \Omega} &= \beta \breve\Omega  - \breve\Omega \phi {\phi ^ \top }\breve\Omega,\ \breve\Omega (0) = {\kappa ^{ - 1}}I.
\end{align*}
}
\begin{figure}[htpb]\centering
	\centerline{\includegraphics[width=0.45\textwidth]{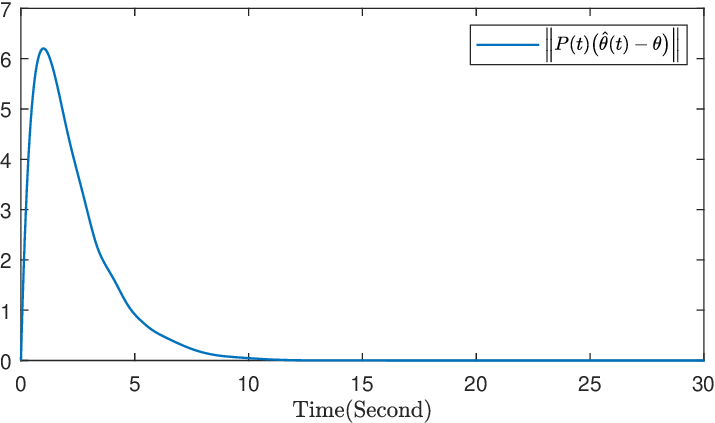}}
	\caption{Subspace parameter learning error when $k=1$.}
	\label{one_fre_noisefree}
\quad
	\centerline{\includegraphics[width=0.45\textwidth]{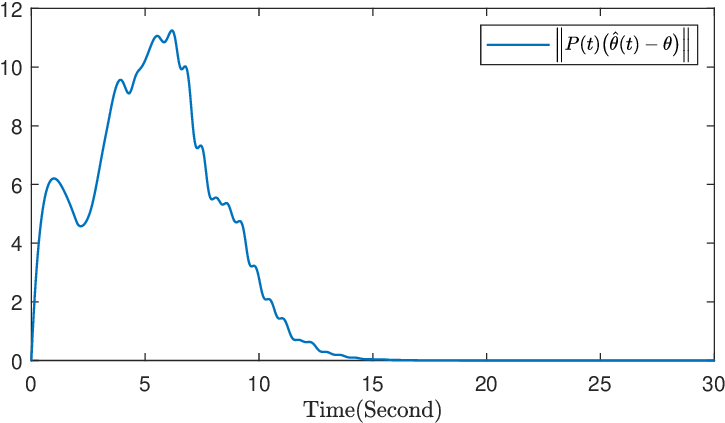}}
	\caption{Subspace parameter learning error when $k=3$.}
	\label{three_fre_noisefree}
\end{figure}  

\begin{figure}[htpb]\centering
	\centerline{\includegraphics[width=0.45\textwidth]{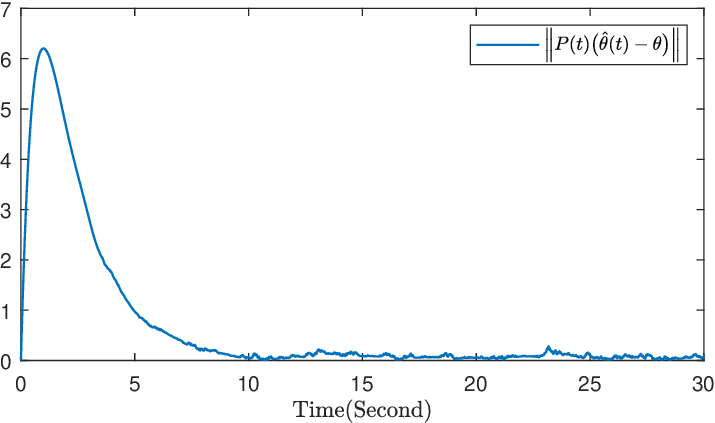}}
	\caption{Subspace parameter learning error when $k=1$ in the presence of white noise.}
	\label{one_fre}
\quad
	\centerline{\includegraphics[width=0.45\textwidth]{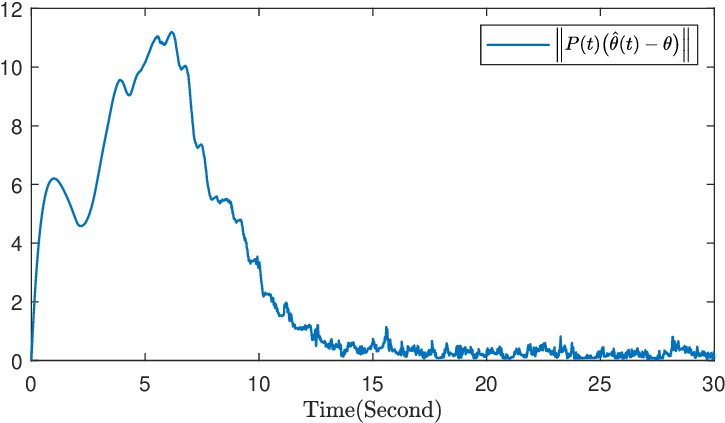}}
	\caption{Subspace parameter learning error when $k=3$ in the presence of white noise.}
	\label{three_fre}
\end{figure}

\begin{figure}[htpb]\centering
	\centerline{\includegraphics[width=0.45\textwidth]{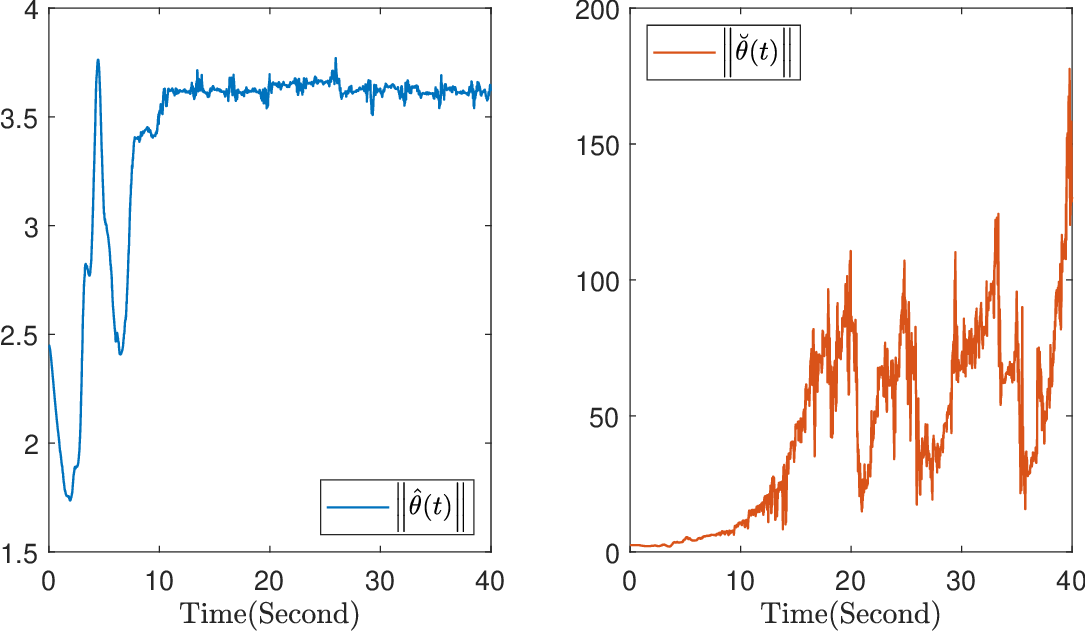}}
	\caption{{Parameter estimate when $k=2$ in the presence of white noise.}}
	\label{two_fre_con1}
\quad
	\centerline{\includegraphics[width=0.45\textwidth]{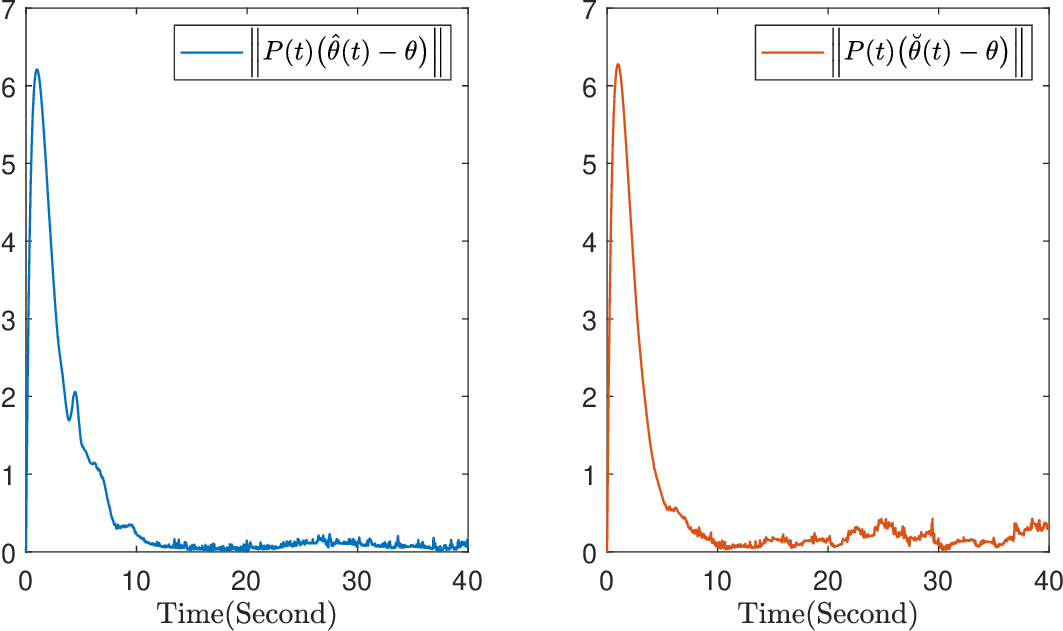}}
	\caption{{Subspace parameter learning error when $k=2$ in the presence of white noise.}}
	\label{two_fre_con2}
\end{figure}

From Figs. \ref{one_fre_noisefree}--\ref{three_fre}, it can be seen that, in the estimated PE subspace ${\mathop{\rm Im}\nolimits} P$, the parameter estimation error $\hat \theta(t)  - \theta$ can converge to zero in the absence of measurement noise, and can converge to a small residual set in the presence of measurement noise. 
%
Here are some further discussions about the simulation results:\\

\textbf{1)} When the frequencies contained in $u$ are not sufficiently rich ($k=1$), the unknown parameters cannot be correctly estimated.  
Nevertheless, at time $t=30s$, we can calculate a full rank factorization $P = {P_d}P_d^{\top}$. According to ${P_d}$, $\hat \theta $, and the relation $P_d^{\top}(t)\left( {\hat \theta (t) - \theta } \right) \approx 0$, the unknown parameters are supposed to satisfy the following two independent constraints:
	\begin{equation*}
		\setlength\arraycolsep{1.5 pt}
		{10^{ - 2}}\left[ {\begin{array}{*{20}{c}}
				{ - 59}&{ - 6}&{59}&{11}&{53}&{ - 12}\\
				0&{ - 62}&0&{ - 54}&{17}&{54}
		\end{array}} \right]\left[ {\begin{array}{*{20}{c}}
				b\\
				f
		\end{array}} \right] \approx {10^{ - 2}}\left[ {\begin{array}{*{20}{c}}
				{ - 52}\\
				{ - 16}
		\end{array}} \right].
	\end{equation*}
{ 
The constraints reflect internal relations among system parameters in real time. The appearance of abnormal relations or relation changes may indicate that the system encounters a poorly conditioned estimation problem. This interesting property could potentially facilitate system monitoring and fault detection.  However, this remains to be investigated.
}

\textbf{2)} When there are three distinct frequencies contained in $u$ ($k=3$), the regressor becomes persistently exciting, and, consequently, it the unknown parameters can be estimated correctly. At time $t=30s$, the matrix $P(t)$ is of full rank, and the values of $\textnormal{col}(b,f)$ are estimated as
	\begin{equation*}
		\textnormal{col}(b,f) \approx \textnormal{col}({1.00}, { - 5.01}, {8.99}, { - 2.51}, { - 11.04}, { - 5.02}).
	\end{equation*}

{\textbf{3)} The standard recursive least squares algorithm is not guaranteed to give a bounded parameter estimate in the presence of noise when the PE condition is not satisfied. In contrast, the proposed algorithm \eqref{hat_theta} can provide a bounded parameter estimate with a better estimation accuracy in the estimated PE subspace.
}

{
\begin{rem}
In view of numerical errors, the singular values of $Q$ are regarded as zero if they are less than a small threshold, based on which $\hat N_u$ is computed in \eqref{QP}.
If the noise added to $y$ is insignificant, the slight mismatch between the measured and the true regressors will not make any difference to $\hat N_u$.
If the noise is significant, it might be impossible to distinguish the PE and non-PE subspace accurately. In this case, by performing the parameter learning algorithm \eqref{hat_theta}, the regressor mismatch is actually taken as part of the term $\varepsilon$ in \eqref{model}.
\end{rem}
}
\subsection{Application 2: Identification for Interconnected Linear Systems}
{In homogeneous multi-agent systems \cite{lewis2013cooperative,Chen2014} and homogeneous large-scale interconnected systems \cite{Yu2017,Xue2021,LiBiao2024}, agents or subsystems share the same dynamic model. The identification problem therein can be formulated as a distributed parameter estimation problem.}
Consider a network of $N$ identical linear time-invariant dynamical systems:
\begin{subequations}\label{NSISO}
	\begin{align}
		{{\dot x}_i} &= F{x_i} + b{u_i} + g\sum\limits_{j = 1}^N {{c_{ij}}h_{(1)}^{\top}{x_j}}, &\label{NSISO_state} \\
		{y_i} &= h_{(1)}^{\top}{x_i},& i = 1,2, \dots,N\ , \label{NSISO_output}
	\end{align}
\end{subequations}
where $x_i \in \mathbb{R}^{n_F}$, $u_i \in \mathbb{R}$, and $y_i \in \mathbb{R}$ are respectively the state, input, and output of the $i$th subsystem, with $F \in \mathbb{R}^{n_F \times n_F}$, $b,\ g,\ {h_{(1)}} \in \mathbb{R}^{n_F}$ all unknown. Let ${c_{ij}}$ be given as either $0$ or $1$, which is known and used to denote the coupling relations among the subsystems. 

It might be difficult or impossible to estimate the unknown parameters by using the input and output information from only one subsystem. So the objective is to design $N$ cooperative estimators for parameter estimation, where the $i$th estimator is in charge of the $i$th subsystem, collecting the information of $u_i$, $y_i$, and $\sum\nolimits_{j = 1}^N {{c_{ij}}{y_j}}$. If $\big( {F,h_{(1)}^{\top}} \big)$ is observable, it imposes no loss of generality to choose the observable canonical form \eqref{Fh} for system identification. Similarly to \eqref{AR}, one can finally arrive at the following algebraic representation of system \eqref{NSISO}:
\begin{equation} \label{Ny_AR}
	{y_i} \!= h_{(1)}^{\top}{{\rm{e}}^{Wt}}{x_i}(0) + h_{(1)}^{\top} \! \left[{\varPi _{yi}}\!\left( {f - w} \right) + {\varPi _{ui}}b + {\varPi _{ci}}g\right],
\end{equation}
where $\setlength\arraycolsep{3 pt} w = \textnormal{col}(
			{w_1}, \dots ,{w_{n_F}})$ is a vector designed such that \eqref{W} is a stable matrix, and $r_{ui}$, $r_{yi}$, and $r_{ci}$ generated by
\begin{subequations}
	\begin{align}
		{{\dot r}_{ui}} &= {W^{\top}}{r_{ui}} + {h_{(1)}}u_i,\ &{r_{ui}}(0) = 0, \label{r_ui}\\
		{{\dot r}_{yi}} &= {W^{\top}}{r_{yi}} + {h_{(1)}}y_i,\ &{r_{yi}}(0) = 0,  \label{r_yi} \\
		{{\dot r}_{ci}} &= {W^{\top}}{r_{ci}} + {h_{(1)}}\sum\nolimits_{j = 1}^N {{c_{ij}}{y_j}},\ &{r_{ci}}(0) = 0,  \label{r_ci}
	\end{align}
\end{subequations}
are bounded signals. The matrices ${\varPi _{ui}}$, ${\varPi _{yi}}$, and ${\varPi _{ci}}$ in \eqref{Ny_AR} are given as
\begin{align*}
 	{\varPi _{ui}} &= H_W^{-1}\col({r_{ui}^{\top}},{r_{ui}^{\top}W},\dots, {r_{ui}^{\top}{W^{{n_F} - 1}}})\\
{\varPi _{yi}} &=H_W^{-1}\col({r_{yi}^{\top}},{r_{yi}^{\top}W},\dots, {r_{yi}^{\top}{W^{{n_F} - 1}}})\\
%
{\varPi _{ci}} &=H_W^{-1}\col({r_{ci}^{\top}},{r_{ci}^{\top}W},\dots, {r_{ci}^{\top}{W^{{n_F} - 1}}}),
%
\end{align*}
where $H_W=\textnormal{col}({h_{(1)}^{\top}}, {h_{(1)}^{\top}} W,\dots, h_{(1)}^{\top}{W^{{n_F} - 1}})$.
The algebraic representation \eqref{Ny_AR} coincides with the regression model \eqref{Nmeasure}, i.e.,
\begin{equation*}
	\setlength\arraycolsep{2 pt}
	\underbrace{{y_i} + h_{(1)}^{\top}{\varPi _{yi}}w}_{z_i} =\underbrace{h_{(1)}^{\top}\left[ {\begin{array}{*{20}{c}}
				{{\varPi _{ui}}}&{{\varPi _{yi}}}&{{\varPi _{ci}}}
		\end{array}} \right]}_{\phi_i^\top} \underbrace{\left[ {\begin{array}{*{20}{c}}
				b\\
				f\\
				g
		\end{array}} \right]}_{\theta} + \underbrace{h_{(1)}^{\top}{{\rm{e}}^{Wt}}{x_i}(0)}_{\varepsilon_i}. 
\end{equation*}

\subsubsection*{\textbf{Numerical Example of Application 2}}
Let $n_F=3$, $N=5$, $b =\textnormal{col}(1, -5, 9)$, $f =\textnormal{col}({-2.5}, {-11}, {-5})$, $g = \textnormal{col}(0, 0, 1)$,
\begin{equation*}
	{c_{ij}} = \left\{ {\begin{array}{*{20}{c}}
			{1,}&{ij \in \left\{ {12,23,34,45,51} \right\}};\\
			{0,}&{\text{otherwise}};
	\end{array}} \right.
\end{equation*}
and choose $\setlength\arraycolsep{3 pt} w = \textnormal{col}(
			{-4}, {-9.25}, {-6.25})$, ${{\hat \theta }_{i0}} = (6-i){\mathbf{1}_{9 \times 1}}$, $\alpha_i = \gamma_i = {\eta _{id}} = i$, $\beta_i =1$, $\delta_i =1$, $\kappa_i = 1$ and ${\eta _{iu}} = 6-i$, $\forall i$. 
\begin{figure}[htbp]
\begin{center}
\tikzstyle{pinstyle} = [pin edge={to-,dashed,very thick,red}]
\begin{tikzpicture}[transform shape]
    \centering%
    \node (3) [circle,draw=blue!20, fill=blue!60, very thick, minimum size=7mm] {\textbf{3}};
     \node (4) [circle,  left =of 3, draw=blue!20, fill=blue!60, very thick, minimum size=7mm] {\textbf{4}};
     \node (5) [circle, above=of 4, draw=blue!20, fill=blue!60, very thick, minimum size=7mm] {\textbf{5}};
   \node (1) [circle, above=of 3, draw=blue!20, fill=blue!60, very thick, minimum size=7mm] {\textbf{1}};
   \node (2) [circle, above right= 0.325cm and 1cm of 3, draw=blue!20, fill=blue!60, very thick, minimum size=7mm] {\textbf{2}};
    \draw[ very  thick,->,  left] (1) --node[pos=0.75,left]{\color{blue}\small}
    (2);
    \draw[ very  thick,->,  left] (1) --node[pos=0.45,below]{\color{blue}\small}
    (5);
    \draw[ very  thick,->,  left] (5) --node[pos=0.45,right]{\color{blue}\small}
    (4);
    \draw[ very  thick,->,   left] (2) --node[pos=0.225,left]{\color{blue}\small}
    (3);
    \draw[ very  thick,->,  left] (4) --node[pos=0.45,above]{\color{blue}\small} (3);
    \draw[ very  thick,->, left] (3) --node[pos=0.45,left]{\color{blue}\small} (1);
\end{tikzpicture}
\end{center}
	\caption{Communication graph.}
	\label{CG}
\end{figure}
Suppose the parameter estimators communicate in a distributed manner as shown in Fig.~\ref{CG}, where the edge weights are all equal to $1$. Take $u_1 = 10\sin(t+1)$, $u_2 = 10\sin(3t+3)$, $u_3 = 10\sin(5t+4)$, $u_4 = 10\sin(3t+3)$, $u_5 = 10\sin(2t+2)$ to obtain the simulation results shown in Fig.~\ref{alle_noisefree}. As in the first example, different white Gaussian noise with unit variance is added to each $y_i$, which leads to the results shown in Fig.~\ref{alle}. 
\begin{figure}[htbp]\centering
	\centerline{\includegraphics[width=0.5\textwidth]{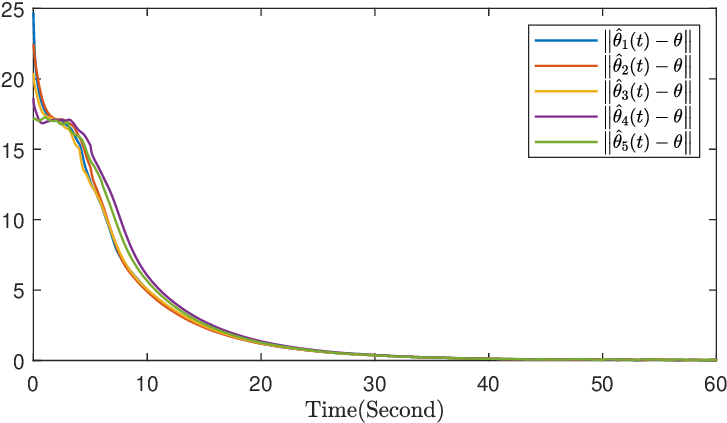}}
	\caption{Distributed parameter learning error at each estimator.}
	\label{alle_noisefree}
\quad
	\centerline{\includegraphics[width=0.5\textwidth]{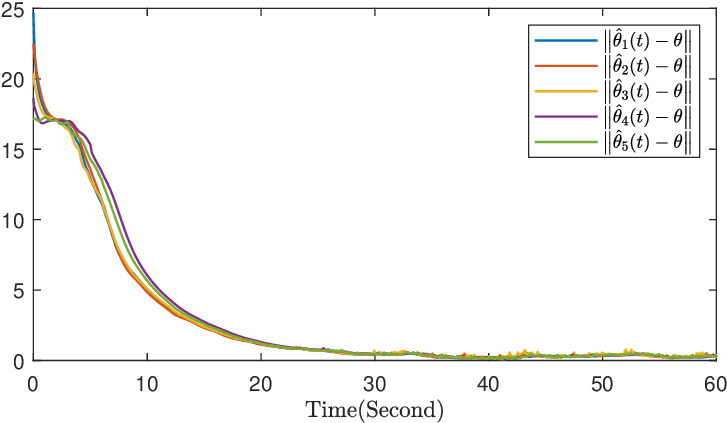}}
	\caption{Distributed parameter learning error at each estimator in the presence of white noise.}
	\label{alle}
\end{figure}
The simulation results show that different estimators can cooperate with each other to estimate all unknown parameters correctly, even though the information collected by any one of the estimators is inadequate for full parameter estimation. 
{
\begin{rem}
In fact, setting $g = 0_{3\times1}$ yields simulation results similar to Figs.~\ref{alle_noisefree} and \ref{alle}. This implies that the estimation of the full parameters is essentially achieved by the consensus protocol, not by the excitation contained in the interconnected term in \eqref{NSISO_state}.
\end{rem}
}

\section{Conclusion} \label{conclution}
A natural extension of the PE condition leads to a PPE condition which induces the definition of {PE} and non-{PE} subspaces.
Despite no prior knowledge of the two subspaces, the proposed algorithm can adaptively distinguish one from the other, and generate an optimal parameter estimate with an exponential rate of convergence. 
Based on that, a distributed parameter learning algorithm is developed, which enables a group of local estimators to cooperate with each other.
The cooperation is established by neighboring communications among local estimators, contributing to a parameter learning task that no local estimator can accomplish alone.
%
The proposed algorithms are helpful in addressing a system identification problem. The simulation results show that the dynamics of the system can be learned from local measurements, none of which satisfies the PE condition.
%
{ One of the future research direction is the development of a finite-time partial persistent excitation condition.
Another meaningful direction is to consider more general model structures that are nonlinear in the parameters.
A linear regression model is only capable of capturing some simple models used in dynamic system identification, such as the finite impulse response (FIR) model and autoregressive with exogenous input (ARX) model. 
Although the ARX model is by far the most widely applied linear dynamic model, it may produce biased and non-consistent estimation in practice \cite{Nelles2023}. In addition, discrete-time versions of the proposed algorithms remain to be developed.
}

\section{Appendix}\label{apendix}
\subsection{Proof of Lemma~\ref{lemma3}}\label{apendix-lemma3}
\begin{proof}
    Consider the following fictitious system:
\begin{equation}\label{lem_sys2}
	{\dot x^*} = {\varUpsilon ^*}{x^*}  + {u^*} ,
\end{equation}
where ${x^*}(0) = 0$. The state difference $\tilde x = x - {x^*}$ between systems \eqref{lem_sys1} and \eqref{lem_sys2} is governed by the dynamics
\begin{align}\label{lem_til_dyn}
		\dot {\tilde x} &= \varUpsilon (t)x  - {\varUpsilon ^*}{x^*}  + u  - {u^*} \nonumber\\
		& = {\varUpsilon ^*}\tilde x  + \left( {\varUpsilon (t) - {\varUpsilon ^*}} \right)\tilde x + \left( {\varUpsilon (t) - {\varUpsilon ^*}} \right){x^*}  + u  - {u^*} .
\end{align}
For any ${\upsilon _0} < \upsilon $, ${\varUpsilon ^*} + \frac{{{\upsilon _0} + \upsilon }}{2}I$ is a stable matrix. In other words, there exists a positive definite matrix $M$ satisfying
\begin{align*}
	M\left( {{\varUpsilon ^*} + (({{{\upsilon _0} + \upsilon }})/{2})I} \right) + {\left( {{\varUpsilon ^*} + (({{{\upsilon _0} + \upsilon }})/{2})I} \right)^{\top}}M < 0.
\end{align*}
Now choose the Lyapunov candidate $\tilde V = {\tilde x^{\top}}M\tilde x$, whose time derivative along the trajectory of \eqref{lem_til_dyn} satisfies
\begin{align*}
		{\dot {\tilde V}}=&\ 2{{\tilde x}^{\top}}M{\varUpsilon ^*}\tilde x + 2{{\tilde x}^{\top}}M\left( {\varUpsilon  - {\varUpsilon ^*}} \right)\tilde x + 2{{\tilde x}^{\top}}M\left( {\varUpsilon  - {\varUpsilon ^*}} \right){x^*} \\
		& + 2{{\tilde x}^{\top}}M\left( {u - {u^*}} \right)\\
		\le&  - \left( {{\upsilon _0} + \upsilon } \right){{\tilde x}^{\top}}M\tilde x + 2{\rho _a}{{\rm{e}}^{ - {\rho _b}t}}\left\| M \right\|{\left\| {\tilde x} \right\|^2}\\
		& + 2\left( {{\rho _a}{{\rm{e}}^{ - {\rho _b}t}}x_m^* + {\rho _c}{{\rm{e}}^{ - {\rho _d}t}}} \right)\left\| {{M^{\frac{1}{2}}}} \right\|\left\| {{M^{\frac{1}{2}}}\tilde x} \right\| \\
		\le& \big( { - {\upsilon _0} - \upsilon  + 2{\rho _a}{{\rm{e}}^{ - {\rho _b}t}}\big\| M \big\|\big\| {{M^{ - 1}}} \big\|} \big){{\tilde x}^{\top}}M\tilde x \\
		& + \frac{{\upsilon  - {\upsilon _0}}}{2}{{\tilde x}^{\top}}M\tilde x + \frac{2{\left( {{\rho _a}x_m^* + {\rho _c}} \right)^2}}{{\upsilon  - {\upsilon _0}}}{{\rm{e}}^{ - 2\rho t}}\left\| M \right\|,
\end{align*}
where $x_m^* = {\sup _{t \ge 0}}\left\| {{x^*}(t)} \right\|$ and $\rho=\min \left\{ {{\rho _b},{\rho _d}} \right\}$. Note that there exists a finite time $t_0$ such that $$4{\rho _a}{{\rm{e}}^{ - {\rho _b}{t}}}\big\| M \big\|\big\| {{M^{ - 1}}} \big\| \le \upsilon  - {\upsilon _0}$$ for all $t \ge {t_0}$. Hence, after time $t_0$, ${\tilde V}$ satisfies
\begin{equation}\label{lem_int_V}
	\tilde V(t) \le {{\rm{e}}^{ - 2{\upsilon _0}(t-t_0)}}\tilde V({t_0}) + \frac{{2{{\left( {{\rho _a}x_m^* + {\rho _c}} \right)}^2}}}{{\upsilon  - {\upsilon _0}}}\left\| M \right\| \varsigma(t), 
\end{equation}
where $\varsigma(t)=\int_{{t_0}}^t {{{\rm{e}}^{ - 2{\upsilon _0}(t - \tau )}}{{\rm{e}}^{ - 2\rho \tau }}{\rm{d}}\tau }$.\\
\textbf{For the case:} $\rho  < {\upsilon _0} < \upsilon$,
\begin{equation} \label{lem_int_V1}
	\varsigma(t)={{\rm{e}}^{ - 2\rho t}}\int_{{t_0}}^t {{{\rm{e}}^{ - 2({\upsilon _0} - \rho )(t - \tau )}}{\rm{d}}\tau } . 
\end{equation}
\textbf{For the case:} $0 < {\upsilon _0} < \rho$,
\begin{equation} \label{lem_int_V2}
	\varsigma(t)={{\rm{e}}^{ - 2{\upsilon _0}t}}\int_{{t_0}}^t {{{\rm{e}}^{ - 2(\rho  - {\upsilon _0})\tau }}{\rm{d}}\tau }. 
\end{equation}
\textbf{For the case:} ${\upsilon _0} = \rho$,
\begin{align} \label{lem_int_V3}
	\varsigma(t) &\le \int_{{t_0}}^t {{{\rm{e}}^{ - 2{\rho _0}(t - \tau )}}{{\rm{e}}^{ - 2\rho \tau }}{\rm{d}}\tau }\nonumber \\
    &= {{\rm{e}}^{ - 2{\rho _0}t}} \! \int_{{t_0}}^t {{{\rm{e}}^{-2({\rho - \rho _0} )\tau }}{\rm{d}}\tau }, 
\end{align}
for any ${\rho _0}$ satisfying $0 < {\rho _0} < \rho $.
Combining \eqref{lem_int_V}, \eqref{lem_int_V1}, \eqref{lem_int_V2}, and \eqref{lem_int_V3} yields that $\left\| {\tilde x}(t)\right\|$ exponentially converges to zero at a decay rate no slower than $\min \left\{ {{\upsilon _0},{\rho _0}} \right\}$. If, in addition, ${u^*}(t)$ vanishes, then system \eqref{lem_sys2} can be analyzed in the same way system \eqref{lem_til_dyn} is analyzed. Taking any ${\rho _{00}}$ satisfying $0 < {\rho _{00}} < {\rho _g}$, it can be proven that $\left\| {{x^*}}(t) \right\|$ exponentially converges to zero at a decay rate no slower than $\min \left\{ {{\upsilon _0},{\rho _{00}}} \right\}$. Given that $\left\| x \right\| \le \left\| {\tilde x} \right\| + \left\| {{x^*}} \right\| $, it is concluded that $\left\| x(t) \right\|$ exponentially converges to zero at a decay rate no slower than $\min \left\{ {{\upsilon _0},{\rho _0},{\rho _{00}}} \right\}$, which completes the proof.
\end{proof}

\subsection{Proof of Theorem~\ref{Theorem1}}\label{proofthe1}
\begin{proof}
\textbf{Step 1:} Prove $0 \le {P(t)} \le I$. Rewrite \eqref{P} as 
\begin{equation}\label{rewriteP}
	\frac{{\rm{d}}\left( {P(t) - I} \right)}{{{\rm{d}}t}} =  - \gamma \left( {P(t) - I} \right) - {\gamma ^2}\bar P(t),
\end{equation}
where $$\bar P(t) = \int_0^t {{{\rm{e}}^{ - \gamma (t - \tau )}}{{\bar N}_u}(\tau )\bar N_u^{\top}(\tau ){\rm{d}}\tau }.$$ Then the solution to \eqref{rewriteP} is 
\begin{equation}\label{soluP}
	P(t) - I = {{\rm{e}}^{ - \gamma t}}\left( {P(0) - I} \right) - {\gamma ^2}\int_0^t {{{\rm{e}}^{ - \gamma (t - \tau )}}\bar P(\tau ){\rm{d}}\tau }.
\end{equation}
It follows from $P(0)=0$ and $\bar P(\tau) \ge 0$ that $P(t) - I \le 0$. Meanwhile, it can be evaluated from \eqref{P} that
\begin{align*}
		\dot P(t) \ge & - \gamma P(t) + (1-{\rm{e}}^{ - \gamma t})\gamma I - {\gamma ^2} \bar P(t)\nonumber\\
		=& - \gamma P(t) + {\gamma ^2} \underbrace{\int_0^t {{{\rm{e}}^{ - \gamma (t - \tau )}}\left(I-{{\bar N}_u}(\tau )\bar N_u^{\top}(\tau ) \right){\rm{d}}\tau }}_{\tilde P(t)},
\end{align*}
%
{where ${{\bar N}_u}\bar N_u^{\top}$ is an orthogonal projection matrix obtained from \eqref{N}, and so $I-{{\bar N}_u}\bar N_u^{\top}\ge 0$.} Consequently, 
\begin{equation*}
	P(t) \ge {{\rm{e}}^{ - \gamma t}} P(0) + {\gamma ^2}\int_0^t {{{\rm{e}}^{ - \gamma (t - \tau )}}\tilde P(\tau ){\rm{d}}\tau }.
\end{equation*}
It follows from $P(0)=0$ and $\tilde P(\tau) \ge 0$ that $P(t) \ge 0$.

\textbf{{Step 2}:} Prove ${\phi^{\top}(t)}{N_{u}} = 0,\ \forall t\ge 0$. Suppose there exist time ${t_u}$ and a column vector $v_u$ in $N_{u}$ such that ${\phi }^{\top}({t_u}){v_u} \ne 0$, then there should be $$v_u^{\top}\phi ({t_u}){\phi }^{\top}({t_u}){v_u} > 0.$$ It combines with the facts that $\phi$ is continuous and $$v_u^{\top}\phi (\tau ){\phi }^{\top}(\tau ){v_u} \ge 0$$ to give $$v_u^{\top}\left(\int_{{t_u}}^{{t_u} + T} {\phi (\tau ){\phi }^{\top}(\tau ){\rm{d}}\tau } \right){v_u} > 0,$$ which contradicts \eqref{PPE3}.

\textbf{{Step 3}:} Prove that ${\bar N_u}\bar N_u^{\top}$ equals $I - {N_d}N_d^{\top}$ in finite time. The solution to \eqref{Q} is  
\begin{equation}\label{SoluQ}
	Q(t) = {{\rm{e}}^{ - \beta t}}Q(0) + \int_0^t {{{\rm{e}}^{ - \beta (t - \tau )}}\phi (\tau ){\phi }^{\top}(\tau ){\rm{d}}\tau }. 
\end{equation}
In view of $Q(0)=0$, the following inequalities hold  
	\begin{align*}
		Q(t) &\ge \int_{t - T}^t {{{\rm{e}}^{ - \beta (t - \tau )}}\phi (\tau ){\phi }^{\top}(\tau ){\rm{d}}\tau } & \\
		&\ge {{\rm{e}}^{ - \beta T}}\int_{t - T}^t {\phi (\tau ){\phi }^{\top}(\tau ){\rm{d}}\tau }, &\forall t \ge T.
	\end{align*}
This leads to
{
\begin{align*}{N^{\top}_{d}}Q(t)N_{d} \ge {{\rm{e}}^{ - \beta T}} k_a I_{n-q},&& \forall t \ge T,
\end{align*}
according to \eqref{PPE1}. Meanwhile, from \eqref{SoluQ}, $Q(0)=0$, and the result of Step 2, it follows that
\begin{align*}
	{N_{u}^{\top}}Q(t)N_{u} = 0, & & \forall t \ge 0.
\end{align*}
Then based on the fact that $\text{rank}({N_{d}})=n-q$, $\text{rank}({N_{u}})=q$, and $N_{d}^\top N_u=0$, an observation is that
\begin{align*}
	{\rm{Ker}}Q(t) = {\rm{Im}}{N_{u}}, & &\forall t \ge T.
\end{align*}
Therefore, applying \eqref{N} gives
\begin{align*}
	{\bar{N}_u(t)}{{\bar{N}}_u}^{\top}(t)= {N_u}{N_u^{\top}} = I - {N_d}{N_d^{\top}},& & \forall t \ge \left\lceil {T/\delta } \right\rceil \delta, 
\end{align*}
where $\left\lceil \cdot \right\rceil$ is the ceiling function. For simplicity, $t \ge T + \delta$ is considered in what follows, instead of $t \ge \left\lceil {T/\delta } \right\rceil \delta$.}

\textbf{{Step 4}:} Complete the proof. Continue the derivation in Step~1 by substituting the result of {Step 3} into $\bar P$:
\begin{align*}
	\bar P(t) = \bar P_a(t) + \bar P_b(t),& & \forall t \ge T + \delta,
\end{align*}
where ${{\bar P}_a}(t) = \int_0^{T + \delta } {{{\rm{e}}^{ - \gamma (t - s)}}{\bar{N}_u}(s){\bar{N}_u}^\top{(s)}{\rm{d}}s} $,
\begin{align*}
		{{\bar P}_b}(t) &= \int_{T + \delta }^t {{{\rm{e}}^{ - \gamma (t - s)}}\left( {I - {N_d}{N_d^\top} } \right){\rm{d}}s } \\
	&= \frac{1}{\gamma }\left( {1 - {{\rm{e}}^{ - \gamma (t - T - \delta )}}} \right)\left( {I - {N_d}{N_d^\top} } \right).
\end{align*}
Further calculations yield
\begin{align}
	\left\| {{{\bar P}_a}(\tau )} \right\| \le \int_0^{T + \delta } {{{\rm{e}}^{ - \gamma (\tau  - s)}}{\rm{d}}s}  =&\ \frac{{{{\rm{e}}^{ - \gamma \tau }}}}{\gamma }\left( {{{\rm{e}}^{\gamma (T + \delta )}} - 1} \right),\nonumber\\
		\left\| {{\gamma ^2}\!\!\!\int_0^t \!{{{\rm{e}}^{ - \gamma (t - \tau )}}{{\bar P}_a}(\tau ){\rm{d}}\tau } } \right\| 
		\le&\ t{{\rm{e}}^{ - \gamma t}}\gamma \left( {{{\rm{e}}^{\gamma (T + \delta )}} - 1} \right),\label{gamma2Pa}\\
		{\gamma ^2}\int_0^t {{{\rm{e}}^{ - \gamma (t - \tau )}}\bar P_b(\tau ){\rm{d}}\tau }
		=&\ \big( {1  - t{{\rm{e}}^{ - \gamma t}}\gamma {{\rm{e}}^{\gamma (T + \delta )}}} \nonumber\\
        &- {{\rm{e}}^{ - \gamma t}}\big)\left( {I - {N_d}{N_d^\top} } \right).\label{int_Pb}
	\end{align} 
Then, it follows from \eqref{soluP} and \eqref{int_Pb} that
	\begin{align}
		P(t) - {N_d}N_d^{\top} &- {{\rm{e}}^{ - \gamma t}}\left( {P(0) - I} \right)\nonumber\\
		= &- {\gamma ^2}\int_0^t {{{\rm{e}}^{ - \gamma (t - \tau )}}\bar P(\tau ){\rm{d}}\tau }  + I - {N_d}{N_d^\top} \nonumber\\
		= &-{\gamma ^2}\int_0^t {{{\rm{e}}^{ - \gamma (t - \tau )}}{{\bar P}_a}(\tau ){\rm{d}}\tau } \nonumber\\
		& + \left( {{{\rm{e}}^{ - \gamma t}} + t{{\rm{e}}^{ - \gamma t}} \gamma {{\rm{e}}^{\gamma (T + \delta )}} } \right)\left( {I - {N_d}{N_d^\top}} \right).\label{P-NN}
	\end{align}
By combining \eqref{gamma2Pa} and \eqref{P-NN}, one can arrive at
\begin{equation*}
	\left\| {P(t) - {N_d}N_d^{\top}} \right\| \le 2{{\rm{e}}^{ - \gamma t}} + t{{\rm{e}}^{ - \gamma t}}\gamma \left( {2{{\rm{e}}^{\gamma (T + \delta )}} - 1} \right). 
\end{equation*}
Note that for any positive $\bar \gamma $ less than $\gamma $,
\begin{equation*}
	\begin{split}
		t{{\rm{e}}^{ - \gamma t}}
		= \int_0^t {{{\rm{e}}^{ - \gamma (t - \tau )}}{{\rm{e}}^{ - \gamma \tau }}{\rm{d}}\tau }
        &\le {{\rm{e}}^{ - \bar \gamma t}}\int_0^t {{{\rm{e}}^{ - (\gamma  - \bar \gamma )(t - \tau )}}{\rm{d}}\tau } \\
		&= {{\rm{e}}^{ - \bar \gamma t}}\frac{{1 - {{\rm{e}}^{ - (\gamma  - \bar \gamma )t}}}}{{\gamma  - \bar \gamma }} 
        \le \frac{{{{\rm{e}}^{ - \bar \gamma t}}}}{{\gamma  - \bar \gamma }},
	\end{split}
\end{equation*}
which leads to
\begin{equation} \label{geTdelta}
	\left\| {P(t) - {N_d}{N_d^{\top}}} \right\| \le \frac{{2\gamma {{\rm{e}}^{\gamma (T + \delta )}} + \gamma  - 2\bar \gamma }}{{\gamma  - \bar \gamma }}{{\rm{e}}^{ - \bar \gamma t}},\ \forall t \ge T + \delta. 
\end{equation}
For the case $0\le t<T + \delta$, it can be obtained from \eqref{soluP} that
	\begin{align}
		\left\| {P(t) - {N_d}{N_d^{\top}}} \right\| &\le {{\rm{e}}^{ - \gamma t}} \!+\! {\gamma ^2}\!\!\!\int_0^t \!{{{\rm{e}}^{ - \gamma (t - \tau )}}\!\left\| {\bar P(\tau )} \right\|\!{\rm{d}}\tau } \! + \!1\nonumber\\
		&\le 2 - t{{\rm{e}}^{ - \gamma t}}\gamma  \le 2.\label{leTdelta}
	\end{align}
According to \eqref{geTdelta} and \eqref{leTdelta},
\begin{equation*}
	\left\| {P(t) - {N_d}{N_d^{\top}}} \right\| \le {\rho _a}{{\rm{e}}^{ - {\rho _b}t}},
\end{equation*}
where $\rho_a= \max \left\{ {2{{\rm{e}}^{\bar \gamma (T + \delta )}},\frac{{2\gamma {{\rm{e}}^{\gamma (T + \delta )}} + \gamma  - 2\bar \gamma }}{{\gamma  - \bar \gamma }}} \right\}$ and $\rho _b=\bar \gamma$, for any positive $\bar \gamma $ less than $\gamma $.
\end{proof}

\subsection{Proof of Theorem~\ref{Theorem2}}\label{proofthe2}
\begin{proof}
\textbf{Step 1:} Find the least squares solution. The least squares solution that minimizes $J$ can be obtained by solving
\begin{align}\label{par_der}
		{\left. {\frac{{\partial J(\vartheta)}}{{\partial \vartheta }}} \right|_{\vartheta  = {\theta ^*}}}  =&  -   \int_0^t \! {{{\rm{e}}^{ - \beta (t - \tau )}}\left( {z(\tau ) - {\phi }^{\top}(\tau ){\theta ^*}(t)} \right)\phi (\tau ){\rm{d}}\tau } \nonumber\\
		& + \alpha{{\rm{e}}^{ - \beta t}}\big( {{\theta ^*}(t) - {{\hat \theta }_0}} \big) \equiv 0
\end{align}
for ${\theta ^*}(t)$. Recall that Step 2 in \ref{proofthe1} has proved ${\phi^{\top}(t)}{N_{u}} = 0$.
Pre-multiplying both sides of \eqref{par_der} by $N_u^{\top}$ gives $N_u^{\top}{\theta ^*}(t) = N_u^{\top}{\hat{\theta}_0}$, which is a necessary condition for the least squares solution ${\theta ^*}(t)$. In other words,
\begin{align} \label{nec_con2}
	{\theta ^*} = {N_d}{N_d^{\top}}{\theta ^*} + {N_u}{N_u^{\top}}{\hat \theta _0},
\end{align}
which is obtained using the fact ${N_d}{N_d^{\top}} + {N_u}{N_u^{\top}} = I$. Meanwhile, both sides of \eqref{par_der} are pre-multiplied by ${N_d^{\top}}$, and, upon substitution of  \eqref{nec_con2} into the resulting expression, with the help of identities $${\phi^{\top}(t)}{N_{u}} = 0\ \ \ \textnormal{and}\ \ \ {N_d}{N_d^{\top}} + {N_u}{N_u^{\top}} = I,$$ the following result can be obtained:
\begin{align*}
\varPsi(t) N_d^{\top}{\theta ^*}(t) = N_d^{\top}{\int_0^t {{{\rm{e}}^{ - \beta (t - \tau )}}z(\tau )\phi (\tau ){\rm{d}}\tau }  + \alpha{{\rm{e}}^{ - \beta t}}N_d^{\top}{{\hat \theta }_0}},
\end{align*}
where $$\varPsi(t) = N_d^{\top}\left(\int_0^t {{{\rm{e}}^{ - \beta (t - \tau )}}\phi (\tau ){\phi }^{\top}(\tau ){\rm{d}}\tau } \right){N_d} + \alpha{{\rm{e}}^{ - \beta t}}I.$$ 
\textbf{For time $t < T$}, there are $\varPsi  > \alpha{{\rm{e}}^{ - \beta T}}I > 0$. \\
\textbf{For time $t \ge T$}, there are
\begin{align*}
		\varPsi(t) \ge&\ N_d^{\top}\left(\int_{t - T}^t {{{\rm{e}}^{ - \beta (t - \tau )}}\phi (\tau ){\phi }^{\top}(\tau ){\rm{d}}\tau }\right) {N_d}\\
		\ge&\ {{\rm{e}}^{ - \beta T}}N_d^{\top}\left(\int_{t - T}^t {\phi (\tau ){\phi }^{\top}(\tau ){\rm{d}}\tau }\right) {N_d}.
\end{align*}
Then, it follows from \eqref{PPE1} that $\varPsi(t)$ is positive definite for all time, and so invertible for all time. Therefore, the least squares solution is given by \eqref{nec_con2} with 
\begin{equation} \label{Nd_theta}
	N_d^{\top}{\theta^*}(t)  = {\varPsi}^{-1}(t) N_d^{\top} \varphi(t),
\end{equation}
where $$\varphi(t)={\int_0^t \!\! {{{\rm{e}}^{ - \beta (t - \tau )}}z(\tau )\phi (\tau ){\rm{d}}\tau }  + \alpha{{\rm{e}}^{ - \beta t}}{{\hat \theta }_0}}.$$

\textbf{Step 2:} Rewrite the least squares solution. Given that the matrix ${N_d}$, and even the number of columns that it contains, are totally unknown, the above least squares solution cannot be used to derive an online algorithm. Instead, the solution needs to be rewritten into an appropriate form. Let
\begin{equation} \label{Psi_alp}
	{\varPsi _\kappa }(t) = \left[ {\begin{array}{*{20}{c}}
			{{N_d}}&{{N_u}}
	\end{array}} \right]\left[ {\begin{array}{*{20}{c}}
			\varPsi(t) &0\\
			0&{\kappa I_q}
	\end{array}} \right]\left[ {\begin{array}{*{20}{c}}
			{N_d^{\top}}\\
			{N_u^{\top}}
	\end{array}} \right],
\end{equation}
where $\kappa$ is a positive real constant. It can be checked that ${\varPsi_\kappa}(t)$ is invertible and
	$$\varPsi_\kappa^{ - 1}(t){N_d}{N_d^{\top}} = {N_d}{\varPsi^{ - 1}(t)}{N_d^{\top}},$$
by exploiting the facts ${N_d^{\top}}{N_d} = I$ and ${N_u^{\top}}{N_d} = 0$. Then it follows from \eqref{varphi} and \eqref{Nd_theta} that
\begin{equation} \label{NdNd_the_sta}
	{N_d}{N_d^{\top}}{\theta ^*}(t) = \varPsi_\kappa^{ - 1}(t) {N_d}{N_d^{\top}}\varphi(t),
\end{equation}
where, according to \eqref{Psi_alp},
\begin{align} \label{psi_alp}
		{\varPsi _\kappa }(t) =&\ {N_d}\varPsi(t) {N_d^{\top}} + \kappa {N_u}{N_u^{\top}}\nonumber\\
		 =&\ {N_d}{N_d^{\top}} \left({\int_0^t {{{\rm{e}}^{ - \beta (t - \tau )}}\phi (\tau ){\phi }^{\top}(\tau ){\rm{d}}\tau }}\right) {N_d}{N_d^{\top}}\nonumber\\
		& + \alpha{\rm{e}}^{ - \beta t}{N_d}{N_d^{\top}} + \kappa \left( {I - {N_d}{N_d^{\top}}} \right).
\end{align}
It must be noted that the matrices ${N_d}$ and ${N_d^{\top}}$ only appear in pairs in the above form. Although ${N_d}$ alone has an unknown number of columns $n-q$, ${N_d}{N_d^{\top}}$ has a known fixed size $n \times n$.

\textbf{Step 3:} Prove the invertibility of $\varOmega$. Let
\begin{equation}\label{hat_psi_alp}
	{\hat \varPsi _\kappa }(t) = P(t)Q(t)P(t)+\alpha{{\rm{e}}^{ - \beta t}}P(t)+ \kappa \left( {I - P}(t) \right).
\end{equation}
From \eqref{Q}, \eqref{psi_alp}, and \eqref{hat_psi_alp}, it follows that
	\begin{align}
		{{{\hat \varPsi }_\kappa } - {\varPsi _\kappa }}=&\ {\left( {P - {N_d}{N_d^{\top}}} \right)QP}+\alpha{{\rm{e}}^{ - \beta t}}\left( {P - {N_d}{N_d^{\top}}} \right)\nonumber\\
		&\ { + {N_d}{N_d^{\top}}Q\left( {P - {N_d}{N_d^{\top}}} \right) }- \kappa \left( {P - {N_d}{N_d^{\top}}} \right).\label{psi_dif}
	\end{align}
Given that $Q$, $P$, and ${N_d}{N_d^{\top}}$ are all bounded, it is clear from Theorem~\ref{Theorem1} that $${\lim _{t \to \infty }}({{\hat \varPsi }_\kappa }(t)-{\varPsi _\kappa }(t)) = 0.$$
Since the roots of a polynomial vary continuously as a function of the coefficients \cite{Harris1987}, the eigenvalues of ${\hat \varPsi _\kappa }$ vary continuously and converge to the eigenvalues of ${\varPsi _\kappa }$ as time goes to infinity. 
Note that the matrix ${\varPsi _\kappa }$ is positive definite. Therefore, there exists a time $t_\kappa$ such that all eigenvalues of ${\hat \varPsi_\kappa }$ remain in the half-plane ${\rm{Re}}(s) > \frac{1}{2}{\lambda _{\min }}({\varPsi _\kappa })$ after time $t_\kappa$. This implies the invertibility of ${\hat \varPsi _\kappa }$ for all time $t>t_\kappa$. 

For $t\le t_\kappa$, the invertibility of ${\hat \varPsi_\kappa }$ is proved as follows. Recall from Theorem~\ref{Theorem1} that both $P$ and $I-P$ are positive semidefinite matrices. For the case $P \ne 0$ and $I - P \ne 0$, there are full-rank factorizations $$P = {P_d}{P_d^{\top}}\quad \textnormal{and} \quad I - P = {P_u}{P_u^{\top}},$$ with $P_d$ and $P_u$ each having full column rank. Then ${\hat \varPsi_\kappa }$ can be rewritten as
\begin{equation*}
	\begin{split}
		{{\hat \varPsi }_\kappa } &= {P_d}{P_d^{\top}}Q{P_d}{P_d^{\top}}+ \alpha{{\rm{e}}^{ - \beta t}}{P_d}{P_d^{\top}} + \kappa {P_u}{P_u^{\top}} \\
		& = \left[ {\begin{array}{*{20}{c}}
				{{P_d}}&{{P_u}}
		\end{array}} \right]\left[ {\begin{array}{*{20}{c}}
				{P_d^{\top}Q{P_d}}+ \alpha{{\rm{e}}^{ - \beta t}}I&0\\
				0&{\kappa I}
		\end{array}} \right]\left[ {\begin{array}{*{20}{c}}
				{{P_d^{\top}}}\\
				{{P_u^{\top}}}
		\end{array}} \right].
	\end{split}
\end{equation*}
With $Q(t) \ge 0$, it is not difficult to verify that $${{P_d^{\top}}(t)Q(t){P_d}}(t)+ \alpha{{\rm{e}}^{ - \beta t}}I>0$$ for time $t\le t_\kappa$.
Meanwhile, it is observed that the matrix $\left[ {\begin{array}{*{20}{c}}
		{{P_d}}&{{P_u}}
\end{array}} \right]$ has full row rank because otherwise it contradicts the fact ${P_d}P_d^{\top} + {P_u}P_u^{\top} = I$.
Then ${\hat \varPsi }_\kappa$ must be positive definite, and therefore invertible for time $t\le t_\kappa$.
The proof for the case $P = 0$ or $I - P = 0$ is straightforward.
From the developments above, it is safe to say that ${\hat \varPsi _\kappa }$ is invertible all the time. Taking the time derivative of $\hat \varPsi _\kappa^{ - 1}$ gives
\begin{align}\label{dyn_psi_inv}
	\dot {\hat \varPsi}_\kappa ^{ - 1} =  - \hat \varPsi _\kappa ^{ - 1}{\dot {\hat \varPsi}_\kappa }\hat \varPsi _\kappa ^{ - 1} 
    =  - \hat \varPsi _\kappa ^{ - 1}\left(R-\beta \hat \varPsi _\kappa\right)\hat \varPsi_\kappa ^{ - 1},
\end{align}
where the second equality can be checked from \eqref{Q}, \eqref{R} and \eqref{hat_psi_alp}. Hence, $\hat \varPsi _\kappa ^{ - 1}$ evolves according to the dynamics \eqref{dyn_psi_inv} with $\hat \varPsi _\kappa ^{ - 1}(0) = {\kappa ^{ - 1}}I$. Due to the existence and uniqueness of a solution to differential equations, comparing \eqref{omega} and \eqref{dyn_psi_inv} yields $\hat \varPsi _\kappa ^{ - 1} = \varOmega $, and therefore $\varOmega $ is invertible.

\textbf{Step 4:} Prove $\big\| {\hat \theta (t) - {\theta ^*}(t)} \big\| \le {\rho _a}{{\rm{e}}^{ - {\rho _b}t}}$. It can be obtained from \eqref{thetad} and \eqref{omega} that
\begin{align} \label{ome_inv_the}
		\frac{{\rm{d}}\left( {{\varOmega ^{ - 1}}{{\hat \theta }_d}} \right)}{{{\rm{d}}t}} =&  - {\varOmega ^{ - 1}}\dot \varOmega {\varOmega ^{ - 1}}{{\hat \theta }_d} + {\varOmega ^{ - 1}}{{\dot {\hat \theta} }_d} \nonumber\\
		=&  - \beta {\varOmega ^{ - 1}}{{\hat \theta }_d} + zP\phi  + \dot P\varphi.
\end{align}
The solution to \eqref{ome_inv_the} is
\begin{equation}\label{dyn_ome_the}
	{\varOmega }(t)^{ - 1}{{\hat \theta }_d}(t) = {{\rm{e}}^{ - \beta t}}{\varOmega ^{ - 1}}(0){{\hat \theta }_d}(0)+ \bar \varphi(t),
\end{equation}
where $$\bar \varphi(t) = \int_0^t {{{\rm{e}}^{ - \beta (t - \tau )}}\left( {z(\tau )P(\tau )\phi (\tau ) + \dot P(\tau )\varphi (\tau )} \right){\rm{d}}\tau }.$$ With \eqref{varphi}, it is not difficult to verify that
\begin{equation*}
	\begin{split}
		\frac{{\rm{d}}\big( {{{\rm{e}}^{ - \beta (t - \tau )}}P(\tau )\varphi (\tau )} \big)}{{{\rm{d}}\tau }}
		 =&\ {{\rm{e}}^{ - \beta (t - \tau )}}\left( {\dot P(\tau )\varphi (\tau ) + \beta P(\tau )\varphi (\tau )} \right)\\
		& + {{\rm{e}}^{ - \beta (t - \tau )}}P(\tau )\left( { - \beta \varphi (\tau ) + z(\tau )\phi (\tau )} \right)\\
		=&\ {{\rm{e}}^{ - \beta (t - \tau )}}\left( {\dot P(\tau )\varphi (\tau ) + z(\tau )P(\tau )\phi (\tau )} \right).
	\end{split}
\end{equation*}
Then a direct calculation gives $\bar \varphi (t) = P(t)\varphi (t)$, which, together with \eqref{dyn_ome_the} and ${{\hat \theta }_d}(0)=0$, leads to ${\hat \theta }_d(t) = \varOmega(t) P(t)\varphi(t)$. 
Now combining it with \eqref{thetau}, \eqref{theta}, \eqref{nec_con2}, \eqref{NdNd_the_sta}, and $\hat \varPsi _\kappa ^{ - 1} = \varOmega $ proved in Step 3, the following expression is obtained
\begin{equation} \label{thehat-thestar}
	\begin{split}
		\hat \theta  - {\theta ^*} &= {{\hat \theta }_d} - {N_d}{N_d^{\top}}{\theta ^*} + {{\hat \theta }_u} - {N_u}{N_u^{\top}}{\theta ^*}\\
		&= \big( {\hat \varPsi _\kappa ^{ - 1}P - \varPsi _\kappa ^{ - 1}{N_d}{N_d^{\top}}} \big)\varphi  + \big( {I - P - {N_u}{N_u^{\top}}} \big){{\hat \theta }_0}.
	\end{split}
\end{equation}
Given that the measurement noise is bounded, the vector $\varphi$ generated by \eqref{varphi} is also bounded. Then from \eqref{psi_dif}, \eqref{thehat-thestar},
\begin{align*}
		\big\| {\hat \varPsi _\kappa ^{ - 1}P - \varPsi _\kappa ^{ - 1}{N_d}N_d^{\top}} \big\| 
		\le&\ \big\| {\hat \varPsi _\kappa ^{ - 1}\big( {P - {N_d}N_d^{\top}} \big)} \big\| \\
       & + \big\| {\big( {\hat \varPsi _\kappa ^{ - 1} - \varPsi _\kappa ^{ - 1}} \big){N_d}N_d^{\top}} \big\|\\
		\le&\ \big\| {\hat \varPsi _\kappa ^{ - 1}} \big\|\big\| {P - {N_d}N_d^{\top}} \big\| + \big\| {\hat \varPsi _\kappa ^{ - 1} - \varPsi _\kappa ^{ - 1}} \big\|,\\
%
		\left\| {\hat \varPsi _\kappa ^{ - 1} - \varPsi _\kappa ^{ - 1}} \right\| =&\ \big\| {\hat \varPsi _\kappa ^{ - 1}\big( {{\varPsi _\kappa } - {{\hat \varPsi }_\kappa }} \big)\varPsi _\kappa ^{ - 1}} \big\| \\
		\le&\ \big\| {\hat \varPsi _\kappa ^{ - 1}} \big\|\big\| {\varPsi _\kappa ^{ - 1}} \big\|\big\| {{{\hat \varPsi }_\kappa } - {\varPsi _\kappa }} \big\|,\\
\left\| {I - P - {N_u}N_u^{\top}} \right\| =& \left\| {P - {N_d}N_d^{\top}} \right\|,
\end{align*}
which lead to
\begin{equation*}
	\begin{split}
	  \big\| \hat \theta (t) -& {\theta ^*}(t) \big\| \le \left\| {P(t) - {N_d}N_d^ \top } \right\| \\
      \times& \left[ {\left( {1 + \left( {2{Q_m} + \left| {\alpha {{\rm{e}}^{ - \beta t}} - \kappa } \right|} \right)\varPsi _{\kappa m}^{ - 1}} \right)\hat \varPsi _{\kappa m}^{ - 1}{\varphi _m} + \left\| {{{\hat \theta }_0}} \right\|} \right],
	\end{split}
\end{equation*}
where 
\begin{align*}
{Q_m} &= \mathop {\max }\limits_{t \ge 0} \left\| {Q(t)} \right\|,\quad~\varPsi _{\kappa m}^{ - 1} = \mathop {\max }\limits_{t \ge 0} \left\| {\varPsi _\kappa ^{ - 1}(t)} \right\|,\\
\hat \varPsi _{\kappa m}^{ - 1} &= \mathop {\max }\limits_{t \ge 0} \left\| {\hat \varPsi _\kappa ^{ - 1}(t)} \right\|,\quad {\varphi _m} = \mathop {\max }\limits_{t \ge 0} \left\| {\varphi (t)} \right\|.
\end{align*}
 According to Theorem~\ref{Theorem1}, the exponential convergence of $\hat \theta(t)  - {\theta ^*}(t)$ follows. In particular, it is obtain that
\begin{equation} \label{inteme_thehat-thestar}
    \big\| {\hat \theta (t) - {\theta ^*}(t)} \big\| \le {\rho _a}{{\rm{e}}^{ - {\rho _b}t}}
\end{equation}
with 
\begin{align*} \rho_a=\, \bigg[ \bigg( 1 + \big( 2{Q_m} + & \alpha + \kappa  \big)\varPsi _{\kappa m}^{ - 1} \bigg)\hat \varPsi _{\kappa m}^{ - 1}{\varphi _m} + \left\| {{{\hat \theta }_0}} \right\| \bigg] \\&\cdot \max \left\{ {2{{\rm{e}}^{\bar \gamma (T + \delta )}},\frac{{2\gamma {{\rm{e}}^{\gamma (T + \delta )}} + \gamma  - 2\bar \gamma }}{{\gamma  - \bar \gamma }}} \right\}\end{align*} and $\rho _b=\bar \gamma$, for any positive $\bar \gamma $ less than $\gamma $.

\textbf{Step 5:} Prove $\big\| N_d^{\top}\big( {\hat \theta(t)  - \theta } \big) \big\| \le {\rho_c}{{\rm{e}}^{ - \frac{\beta}{2}t}}$.
Recall the expression for $\varPsi$ from Step 1:
\begin{equation*}
	\varPsi(t) = {N_d^{\top}}\left(\int_0^t {{{\rm{e}}^{ - \beta (t - \tau )}}\phi (\tau ){\phi }^{\top}(\tau ){\rm{d}}\tau }\right) {N_d} + \alpha{{\rm{e}}^{ - \beta t}}I.
\end{equation*}
Take the time derivative of both sides to yield
\begin{equation}\label{dyn_psi}
	\dot \varPsi =  - \beta \varPsi + {N_d^{\top}}\phi (t){\phi }^{\top}(t){N_d}.
\end{equation}
Now let $\tilde \theta _d^* = {N_d^{\top}}{\theta ^*} - {N_d^{\top}}\theta$. It follows from \eqref{varphi}, \eqref{Nd_theta}, \eqref{dyn_psi}, and \eqref{model} with $\varepsilon=0$ that
\begin{equation*}
	\begin{split}
		\dot {\tilde \theta}_d^* =&  - {\varPsi ^{ - 1}}\dot \varPsi {\varPsi ^{ - 1}}{N_d^{\top}}\varphi  + {\varPsi ^{ - 1}}N_d^{\top}\dot \varphi \\
        =&\   \beta {\varPsi ^{ - 1}} {N_d^{\top}}\varphi- {\varPsi ^{ - 1}}{N_d^{\top}}\phi (t){\phi }^{\top}(t){N_d} {\varPsi ^{ - 1}}{N_d^{\top}}\varphi\\
        &- \beta{\varPsi ^{ - 1}}N_d^{\top}  \varphi  + {\varPsi ^{ - 1}}N_d^{\top}\phi{\phi }^{\top}(t)\theta \\
         =&  - {\varPsi ^{ - 1}}{N_d^{\top}}\phi (t){\phi }^{\top}(t){N_d} \underbrace{{\varPsi ^{ - 1}}{N_d^{\top}}\varphi}_{{N_d^{\top}}{\theta ^*} (\textnormal{see}~\eqref{Nd_theta})}  + {\varPsi ^{ - 1}}N_d^{\top}\phi{\phi }^{\top}(t)\theta \\
          =&- {\varPsi ^{ - 1}}{N_d^{\top}}\phi (t){\phi }^{\top}(t){N_d}\underbrace{({N_d^{\top}}{\theta ^*}-{N_d^{\top}}\theta)}_{\tilde \theta _d^*}\\
         &+ {\varPsi ^{ - 1}}{N_d^{\top}}\phi (t){\phi }^{\top}(t){N_u}{N_u^{\top}}\theta.
	\end{split}
\end{equation*}
Then, by exploiting ${\phi^{\top}(t)}{N_{u}} = 0$ and ${N_d}{N_d^{\top}} + {N_u}{N_u^{\top}} = I$, it can be obtained that
\begin{equation} \label{dyn_til_the_sta}
	\dot {\tilde \theta}_d^* =  - {\varPsi ^{ - 1}}{N_d^{\top}}\phi (t){\phi }^{\top}(t){N_d}\tilde \theta _d^*.
\end{equation}
In order to prove the exponential convergence of ${\tilde \theta}_d^*$, choose a Lyapunov candidate
\begin{equation}\label{Lya}
	V(\tilde{\theta}_d^*)= \tilde \theta _d^{*{\top}}\varPsi \tilde \theta _d^*,
\end{equation}
where $\varPsi$ is positive definite according to Step 1.
Take the time derivative of \eqref{Lya} along the trajectories of \eqref{dyn_psi} and \eqref{dyn_til_the_sta} to give
\begin{align*}
		\dot V &= \tilde \theta _{d}^{*{\top}}\dot \varPsi \tilde \theta _{d}^* + 2\tilde \theta _{d}^{*{\top}}\varPsi \dot {\tilde \theta}_{d}^* \\
		&=  - \beta V - \tilde \theta _d^{*{\top}}N_d^{\top}\phi (t){\phi }^{\top}(t){N_d}\tilde \theta _d^* \le - \beta V.
\end{align*}
This, together with the fact that $$V(t) \ge {\inf\nolimits_{t \ge 0}}\left( {{\lambda _{\min }}\left( {\varPsi (t)} \right)} \right){\tilde \theta _d^{*{\top}} \tilde \theta _d^*},$$ leads to
\begin{equation}\label{til_the_d_con}
	\left\| {\tilde \theta _d^*} \right\| =\left\|N_d^{\top}({\theta ^*} - \theta)\right\| \le {{\rm{e}}^{ - \frac{\beta }{2}t}}\frac{\sqrt {{V(0)}}}{{\sqrt {{{\inf }_{t \ge 0}}\big( {{\lambda _{\min }}\big({\varPsi (t)} \big)} \big)}}}, 
\end{equation}
which implies the exponential convergence of $N_d^{\top}\left( {{\theta ^*}(t) - \theta } \right)$ at a decay rate no slower than $\beta /2$ with respect to time $t$. Given that
\begin{equation*}
	N_d^{\top}\big( {\hat \theta  - \theta } \big) = N_d^{\top}\big( {\hat \theta  - {\theta ^*}} \big) + N_d^{\top}\left( {{\theta ^*} - \theta } \right),
\end{equation*}
and combining \eqref{inteme_thehat-thestar} and \eqref{til_the_d_con}, the exponential convergence of $N_d^{\top}\big( {\hat \theta - \theta } \big)$ follows, i.e.,
\begin{equation*}
    \big\| N_d^{\top}\big( {\hat \theta(t)  - \theta } \big) \big\| \le {\rho_c}{{\rm{e}}^{ - \frac{\beta}{2}t}}
\end{equation*}
with \begin{align*}\rho_c=& \left[ {\left( {1 + \left( {2{Q_m} + {\alpha + \kappa }} \right)\varPsi _{\kappa m}^{ - 1}} \right)\hat \varPsi _{\kappa m}^{ - 1}{\varphi _m} + \left\| {{{\hat \theta }_0}} \right\|} \right]\\
&\cdot \max \left\{ {2{{\rm{e}}^{\bar \gamma (T + \delta )}},\frac{{2\gamma {{\rm{e}}^{\gamma (T + \delta )}}
+ \gamma  - 2\bar \gamma }}{{\gamma  - \bar \gamma }}} \right\}\\& + \frac{{\sqrt {V(0)}}}{\sqrt {{{{\inf }_{t \ge 0}}\big( {{\lambda _{\min }}\big({\varPsi (t)} \big)} \big)}}},\end{align*} 
for any $\bar \gamma$ and $\gamma$ satisfying $\beta /2 \le \bar \gamma < \gamma$.
\end{proof}

\subsection{Proof of Theorem~\ref{Theorem3}}\label{proofthe3}
\begin{proof}
\textbf{Step 1:} Prove Property~\ref{property1}. Consider the following relations:
	\begin{align}\label{int_dif1}
		\big\|  {N_U^{\top}}\big( {\cal L} \otimes& I \big){N_U}  - {N_U^{\top}}{P_U}\left( {{\cal L} \otimes I} \right){P_U}{N_U}  \big\| \nonumber\\
		\le& \left\| {\left( {{N_U^{\top}}{P_U} - N_U^{\top}} \right)\left( {{\cal L} \otimes I} \right){P_U}{N_U}} \right\|  \nonumber\\
		& + \left\| {{N_U^{\top}}\left( {{\cal L} \otimes I} \right)\left( {{P_U}{N_U} - {N_U}} \right)} \right\| \nonumber\\
		\le &\left\| {N_U^{\top}}{P_U} - {N_U^{\top}} \right\|\left\| {{\cal L} \otimes I} \right\|\left( {\left\| {{P_U}{N_U}} \right\| + \left\| {{N_U}} \right\|} \right)
	\end{align}
and 
\begin{align} \label{int_dif2}
	\left\|{N_U^{\top}}{P_U} - N_U^{\top}\right\| &= \left\|{N_U^{\top}}{P_D}\right\| = \left\|{N_U^{\top}}{P_D} - {N_U^{\top}}N_D {N_D^{\top}} \right\| \nonumber\\
	&\le \left\| {P_D} - N_D N_D^{\top} \right\|.
\end{align}
Given that ${\cal L}$, $P_U(t)$, $N_U$, and $H_U$ are all bounded, it comes from \eqref{int_dif1}, \eqref{int_dif2}, and Theorem~\ref{Theorem1} that there exist $\rho_a,\ \rho_b>0$ such that
\begin{align}\label{HNPLPN-HNLN}
	\big\| {H_U}{N_U^{\top}}{P_U(t)}( {\cal L} &\otimes I ){P_U(t)}{N_U}  \nonumber\\
	&  \left.-{H_U}{N_U^{\top}}\left( {{\cal L} \otimes I} \right){N_U} \right\| \le {\rho _a}{{\rm{e}}^{ - {\rho _b}t}},
\end{align}
where $\rho_b$ can be made arbitrarily large by increasing $\gamma_i$. Suppose there exists a nonzero vector ${\bar v_u} \in  \cap _{i = 1}^N{\mathop{\rm Im}\nolimits} {N_{iu}}$, then according to \eqref{varPhiiab} and \eqref{PPE3}, {$\sum\nolimits_{i = 1}^N {\bar v_u^{\top}{\varPhi_{i}}{{\bar v}_u}}  = 0$,}
which contradicts the complementary PPE condition. Hence, $ \cap _{i = 1}^N{\mathop{\rm Im}\nolimits} {N_{iu}} = \left\{ 0 \right\}$. Then by Lemma~\ref{lemma1}, there exists a positive definite matrix $\varXi_0  = {\rm{diag}}{\left( \xi_1,\dots, \xi _N\right)}$ such that 
\begin{equation*}
{N_U^{\top}}\left[ {\left( {{\varXi _0}{\cal L} + {{\cal L}^{\top}}{\varXi _0}} \right) \otimes I} \right]{N_U} > 0,
\end{equation*}
which implies that the inequality
\begin{equation*}
	\varXi {H_U}{N_U^{\top}}\big( {{\cal L} \otimes I} \big){N_U} + {N_U^{\top}}\left( {{{\cal L}^{\top}} \otimes I} \right){N_U}{H_U}\varXi > 0
\end{equation*}
has a positive definite solution $$\varXi  = {\rm{diag}}{\big( {{\xi _1}\eta _{1u}^{ - 1}{I_{{q_1}}}},\dots, {{\xi_N
}\eta _{Nu}^{ - 1}{I_{{q_N}}}}\big)}.$$ 
Therefore, according to Lemma~\ref{lemma2}, $- {H_U}{N_U^{\top}}\left( {{\cal L} \otimes I} \right){N_U}$ is a stable matrix. Moreover, its eigenvalues can be placed arbitrarily far from the imaginary axis, by increasing $\eta _{iu}$.

\textbf{Step 2:} Prove the boundedness of ${\tilde \theta _U}$. According to \eqref{til_def} and \eqref{dyn_til_the}, the overall error dynamics system can be written as
	\begin{align}
		{{\dot {\tilde \theta}}_U} =& \underbrace{(- {H_D}{P_D(t)} - {H_U}{P_U(t)}\left( {{\cal L} \otimes I} \right){P_U(t)})}_{\varLambda_a(t)} {{\tilde \theta }_U}\nonumber \\
		& \underbrace{- {H_U}{P_U(t)}\left( {{\cal L} \otimes I} \right){P_D(t)}{{\tilde \theta }_D}- {H_D}{P_D(t)}\left( {{\mathbf{1}_N} \otimes \theta } \right) }_{\varLambda_b(t)}\nonumber \\
		=&\ \varLambda_a(t){{\tilde \theta }_U} + \varLambda_b(t),\label{ove_dyn_U}
	\end{align}
It can be proved in the same way as in Step 1 that there exist $\rho_c,\ \rho_d>0$ such that
	$\left\| \varLambda_a(t) -\varLambda_a^* \right\| \le {\rho _c}{{\rm{e}}^{ - {\rho _d}t}}$,
where
\begin{equation*} 
	\varLambda_a^* = \setlength\arraycolsep{3 pt}
	-\left[ {\begin{array}{*{20}{c}}
			{{N_D}}&{{N_U}}
	\end{array}} \right] \! \left[ {\begin{array}{*{20}{c}}
			{{H_D}}&0\\
			0&{{H_U}{N_U^{\top}}\left( {{\cal L} \otimes I} \right){N_U}}
	\end{array}} \right] \! \left[ {\begin{array}{*{20}{c}}
			{{N_D^{\top}}}\\
			{{N_U^{\top}}}
	\end{array}} \right].
\end{equation*}
It follows from ${H_D} > 0$, Step 1, and the orthogonality of $\left[ {\begin{array}{*{20}{c}}
		{{N_D}}&{{N_U}}
\end{array}} \right]$ that $\varLambda_a^*$ is stable. At the same time, based on \eqref{hatthetaid-star} and Theorem~\ref{Theorem1}, it can be verified that there exist $\rho_e,\ \rho_f>0$ such that
\begin{align*}
	\left\| \varLambda_b(t) -\varLambda_b^*(t) \right\| \le {\rho _e}{{\rm{e}}^{ - {\rho _f}t}},
\end{align*}
where 
$$\varLambda_b^*(t)=- {H_U}{N_U}{N_U^{\top}}\left( {{\cal L} \otimes I} \right){N_D}{N_D^{\top}}\left( {\theta _I^*- {\theta _I}} \right) - {H_D}{N_D}{N_D^{\top}}{\theta _I}$$
with $\theta _I^* = {\rm{col}}{\left( {\theta_1^*},\dots, {\theta_N^*}\right)}$ and ${\theta _I} = {\mathbf{1}_N} \otimes \theta$.
The signal $\varLambda_b^*(t)$ is uniformly bounded since the measurement noise $\varepsilon_i$ is bounded and the parameter $\theta$ is constant as formulated in Section \ref{PF}. By applying Lemma~\ref{lemma3} to system \eqref{ove_dyn_U}, it can be concluded that ${\tilde \theta _U}$ is uniformly bounded.

\textbf{Step 3:} Complete the proof. Based on Step 2, it follows from \eqref{til_def} that ${\hat \theta _U}$ is also uniformly bounded. 
Meanwhile, according to \eqref{int_dif2}, Theorem~1, and the relation
\begin{align*} 
	\left\|{P_U}N_D\right\| &= \left\|\left(I-{P_D}\right)N_D\right\| \nonumber\\
	&= \left\|\left(N_D {N_D^{\top}}-P_D\right)N_D\right\| 
	\le \left\| {P_D} - N_D {N_D^{\top}} \right\|,
\end{align*}
there exist $\rho_g,\ \rho_h>0$ such that
\begin{align} \label{PNNPexp}
	\left\| {{P_U}{N_D}} \right\| \le {\rho _g}{{\rm{e}}^{ - {\rho _h}t}}\ \text{and}\ \left\|{N_U^{\top}{P_D}} \right\| \le {\rho _g}{{\rm{e}}^{ - {\rho _h}t}},
\end{align}
where $\rho_h$ can be made arbitrarily large by increasing $\gamma_i$.
According to \eqref{hatthetaid-star}, \eqref{int_dif2}, and Theorem 1, there exist $\rho_l,\ \rho_m>0$ such that
\begin{subequations}\label{HNPLPtheta}
	\begin{align}
		\left\| {N_U^{\top}}{P_U}(t) - {N_U^{\top}} \right\| &\le {\rho _l}{{\rm{e}}^{ - {\rho _m}t}}\\
		\left\|{P_D}(t){{\tilde \theta }_D}(t) - {N_D}{N_D^{\top}}\left( {\theta _I^*(t) - {\theta _I}} \right) \right\| &\le {\rho _l}{{\rm{e}}^{ - {\rho _m}t}},\label{PDtiltheD-NNthestar-theI}
	\end{align}
\end{subequations}
where $\rho_m$ can be made arbitrarily large by increasing $\gamma_i$. 
Due to the boundedness of ${\hat \theta _U}$ and ${\tilde \theta _U}$, it follows from \eqref{PNNPexp} and \eqref{HNPLPtheta} that there exist $\rho_o,\ \rho_p>0$ such that
\begin{align} \label{Lamdacexp}
	\left\| \varLambda_c(t) -\varLambda_c^*(t) \right\| \le {\rho _o}{{\rm{e}}^{ - {\rho _p}t}},
\end{align}
where $\rho_p$ can be made arbitrarily large by increasing $\gamma_i$, and 
\begin{align*}
	\varLambda_c(t)=&- {H_U}{N_U^{\top}}{P_U}(t)\left( {{\cal L} \otimes I_n} \right){P_U}(t){N_D}{N_D^{\top}}{{\tilde \theta }_U}(t) \\
	& - {H_U}\! {N_U^{\top}}{P_U}(t)\! \left( {{\cal L} \otimes I_n} \right)\!{P_D}(t){{\tilde \theta }_D}(t)\\
    & - {H_D} {N_U^{\top}}{P_D}(t){{\hat \theta }_U}(t),\\
	\varLambda_c^*(t) =& - {H_U}\! {N_U^{\top}}\! \left( {{\cal L} \otimes I_n} \right)\!{N_D}{N_D^{\top}}\left( {\theta _I^*(t) - {\theta _I}} \right).
\end{align*}
With \eqref{Nthestar-Nthe_i}, \eqref{HNPLPN-HNLN}, and \eqref{Lamdacexp}, by applying Lemma~\ref{lemma3} to system \eqref{ove_dyn}, it follows that there exist $\rho_q,\ \rho_r>0$ such that
\begin{align} \label{NUtiltheU-thitheUstar}
	\left\| {N_U^{\top}}{{\tilde \theta }_U}(t) -\tilde \theta _U^*(t) \right\| \le {\rho _q}{{\rm{e}}^{ - {\rho _r}t}},
\end{align}
where $\rho_r$ can be made arbitrarily large by increasing $\gamma_i$ and $\eta_{iu}$. In particular, for any $\rho_t < \mathop {\min }\limits_{i\in \mathcal{N}} \left\{ {{\beta _i}/2} \right\}$, there exists $\rho_s > 0$ such that
\begin{align} \label{NUtiltheUonly}
	\left\| {N_U^{\top}}{{\tilde \theta }_U}(t) \right\| \le {\rho _s}{{\rm{e}}^{ - {\rho _t}t}},
\end{align}
in the noise-free case $\varepsilon_i (t) \equiv 0$.
According to \eqref{til_rel} and \eqref{refsys_cha},
\begin{align} 
{{\tilde \theta }_I} =&\ {P_D}{{\tilde \theta }_D} + \left( {{N_D}{N_D^\top}  - {P_D}} \right){{\tilde \theta }_U} + {N_U}{N_U^\top} {{\tilde \theta }_U}\label{tiltheI_rewrite}\\
	{{\tilde \theta }_I} - \tilde \theta _I^* =&\ {P_D}{{\tilde \theta }_D} - {N_D}\tilde \theta _D^* + \left( {{N_D}{N_D^\top}  - {P_D}} \right){{\tilde \theta }_U} \nonumber\\
    &+ {N_U}\left( {{N_U^\top} {{\tilde \theta }_U} - \tilde \theta _U^*} \right)\label{diff_tiltheI-Istar}.
\end{align}
Then combining \eqref{Nthestar-Nthe_i}, \eqref{PDtiltheD-NNthestar-theI}, \eqref{NUtiltheU-thitheUstar}, \eqref{NUtiltheUonly}, \eqref{tiltheI_rewrite}, \eqref{diff_tiltheI-Istar}, and following Theorem~\ref{Theorem1}, the results of Theorem~\ref{Theorem3} can be proven.
\end{proof}

\subsection{Algebraic representation of system \eqref{SISO}}\label{ARS-repre}
The algebraic representation dates back to \cite{Kreisselmeier1977}. It is derived and presented here in a more concise way. Consider the fictitious system
\begin{align} \label{Pi_u}
	{{\dot \varPi }_u} = W{\varPi _u} + {I_{{n_F}}}u,& & {\varPi _u}(0) = 0,
\end{align}
where $W \in \mathbb{R}^{n_F \times n_F}$ has the form \eqref{W}, $u \in \mathbb{R}$ is the same as that in \eqref{SISO}, and ${\varPi _u} \in \mathbb{R}^{n_F \times n_F}$ is the state.
In the frequency domain, systems \eqref{Pi_u} and \eqref{r_u} can be expressed as
\begin{subequations}\begin{align}
	{\varPi _u}(s) &= {\left( {sI - W} \right)^{ - 1}}{I_{{n_F}}}u(s), \label{fre_Pi_u}\\
	{r_u}(s) &= {\left( {sI - {W^{\top}}} \right)^{ - 1}}{h_{(1)}}u(s), \label{fre_r_u}
\end{align}\end{subequations}
Let ${h_{(i)}}$ denote the $i$th column of $I_{n_F}$. According to
\begin{align*}
	h_{(1)}^{\top}{\left( {sI - W} \right)^{ - 1}}{h_{(i)}} = h_{(i)}^{\top}{\left( {sI - {W^{\top}}} \right)^{ - 1}}{h_{(1)}}, & &\forall i \in \left\{ {1, \dots ,{n_F}} \right\},
\end{align*}
it can be obtained from \eqref{fre_r_u} and \eqref{fre_r_u} that
\begin{equation*}
	\setlength\arraycolsep{3 pt}
	h_{(1)}^{\top}{\varPi _u}(s) = \underbrace{\left[ {\begin{array}{*{20}{c}}
			{h_{(1)}^{\top}{r_u}(s)}& \cdots &{h_{({n_F})}^{\top}{r_u}(s)}
	\end{array}} \right]}_{ r_u^{\top}(s)}.
\end{equation*}
Likewise, in light of the fact that
\begin{equation*}
	\begin{split}
		h_{(1)}^{\top}{W^{j - 1}}{\left( {sI - W} \right)^{ - 1}}{h_{(i)}} &= h_{(1)}^{\top}{\left( {sI - W} \right)^{ - 1}}{W^{j - 1}}{h_{(i)}} \\
		&= h_{(i)}^{\top}{\big( {{W^{j - 1}}} \big)^{\top}}{\big( {sI - {W^{\top}}} \big)^{ - 1}}{h_{(1)}},
	\end{split}
\end{equation*}
$\forall i,j \in \left\{ {1, \cdots ,{n_F}} \right\}$, the following expression can be obtained:
	\begin{align}\label{hWPi}
		h_{(1)}^{\top}\!{W^{j - 1}}{\varPi _u} =& \left[\! {\setlength\arraycolsep{3 pt}\begin{array}{*{20}{c}}
				{h_{(1)}^{\top}\!{{\left( {{W^{j - 1}}} \right)\!}^{\top}}\!{r_u}}& \cdots\! &{h_{({n_F})}^{\top}\!{{\left( {{W^{j - 1}}} \right)\!}^{\top}}\!{r_u}}
		\end{array}} \!\!\right] \nonumber\\
		=&\ r_u^{\top}{W^{j - 1}}, \quad \quad  \quad \forall j \in \big\{ {1, \cdots ,{n_F}} \big\}
	\end{align}
Since $\big( {W,h_{(1)}^{\top}} \big)$ is observable, the following matrix is invertible
\begin{equation*}
\textnormal{col}\Big({h_{(1)}^{\top}},{h_{(1)}^{\top}W},\dots, {h_{(1)}^{\top}{W^{{n_F} - 1}}}\Big).
\end{equation*}
 Then, \eqref{hWPi} leads to \eqref{Pi_u_mat}, which means ${\varPi _u}$ generated by \eqref{Pi_u} can be expressed in terms of $r_u$ generated by \eqref{r_u}.

In the same way as above, it can be verified that $\varPi_y \in \mathbb{R}^{n_F \times n_F}$ generated by the fictitious system
\begin{align} \label{Pi_y}
	{{\dot \varPi }_y} = W{\varPi _y} + {I_{{n_F}}}y,& & {\varPi _y}(0) = 0
\end{align}
can be expressed in terms of $r_y$ generated by \eqref{r_y}, and the expression is \eqref{Pi_y_mat}. 
After rewriting \eqref{SISO} as
\begin{equation*}
	\begin{split}
		\dot x &= Wx + \left( {F - W} \right)x + bu = Wx + \left( {f - w} \right)y + bu,
	\end{split}
\end{equation*}
it is straightforward from \eqref{SISO} that 
	\begin{align}\label{y_ar}
		y(t) =&\ h_{(1)}^{\top}{{\rm{e}}^{Wt}}x(0) + h_{(1)}^{\top}\int_0^t {{{\rm{e}}^{W(t - \tau )}}\left( {f - w} \right)y(\tau ){\rm{d}}\tau } \nonumber \\
		&  + h_{(1)}^{\top}\int_0^t {{{\rm{e}}^{W(t - \tau )}}bu(\tau ){\rm{d}}\tau }. 
	\end{align}
Finally, combining \eqref{y_ar} with \eqref{Pi_u} and \eqref{Pi_y} gives the algebraic representation \eqref{AR}.

\footnotesize
\bibliographystyle{ieeetr}
\bibliography{myref}
\vspace{-35pt}
\begin{IEEEbiography}[{\includegraphics[height=1.25in, clip,keepaspectratio]{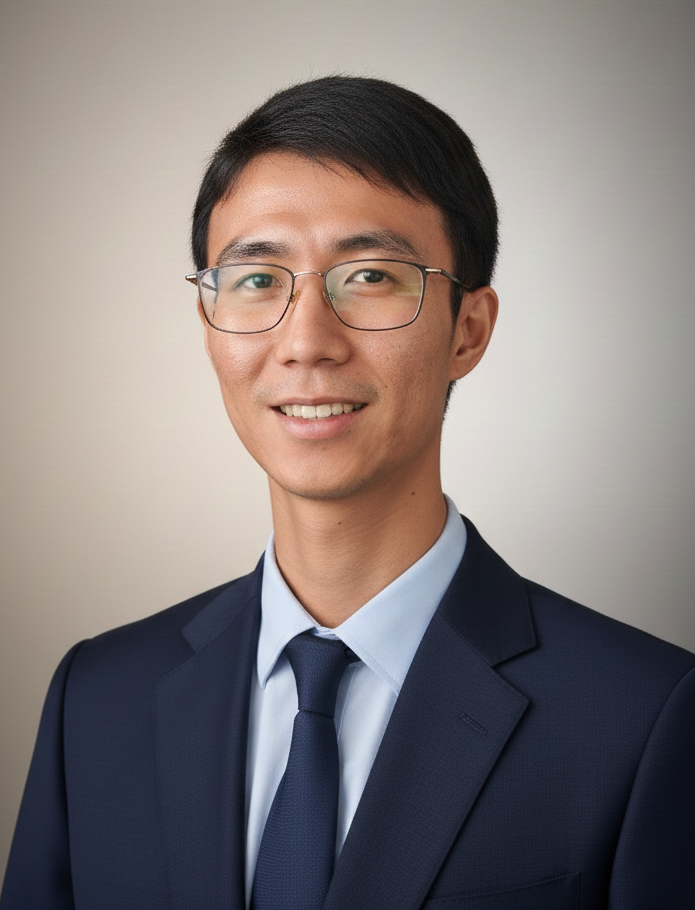}}]{Ganghui Cao} received the B.E. degree in Flight Vehicle Design and Engineering from Harbin Engineering University, China, in 2020. He is currently pursuing a Ph.D. degree with the Department of Mechanics and Engineering Science, College of Engineering, Peking University, China. Since November 2023, he has been a visiting Ph.D. candidate with the KIOS Research and Innovation Center of Excellence at the University of Cyprus. His research interests include distributed estimation and multi-agent systems.
\end{IEEEbiography}
\vspace{-15pt}
 \begin{IEEEbiography}[{\includegraphics[height=1.25in, clip,keepaspectratio]{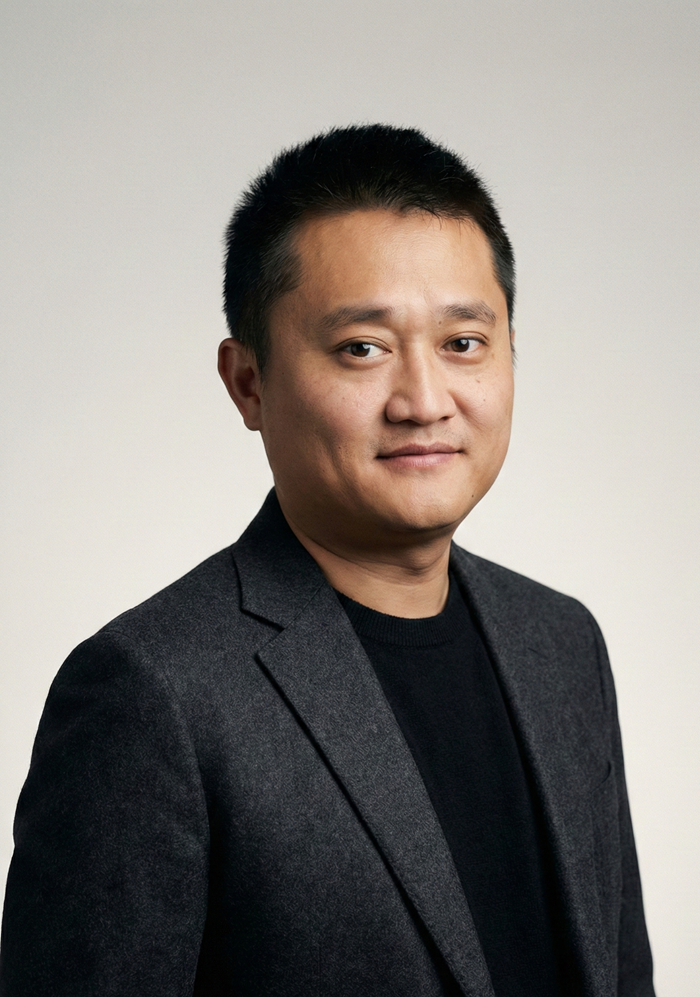}}]{Shimin Wang} (Senior Member, IEEE) received his Ph.D. from the Chinese University of Hong Kong. He is a recipient of the NSERC Postdoctoral Fellowship, the Best Poster Award at the 2024 Nonlinear Systems and Control Conference, and Best Conference Paper Awards at the 2025 IEEE International Conference on Unmanned Systems and 2018 IEEE International Conference on Information and Automation. His research bridges data science, embodied AI, and machine learning to build advanced intelligent autonomous systems.\end{IEEEbiography}  

\vspace{-15pt}
\begin{IEEEbiography}[{\includegraphics[height=1.25in, clip,keepaspectratio]{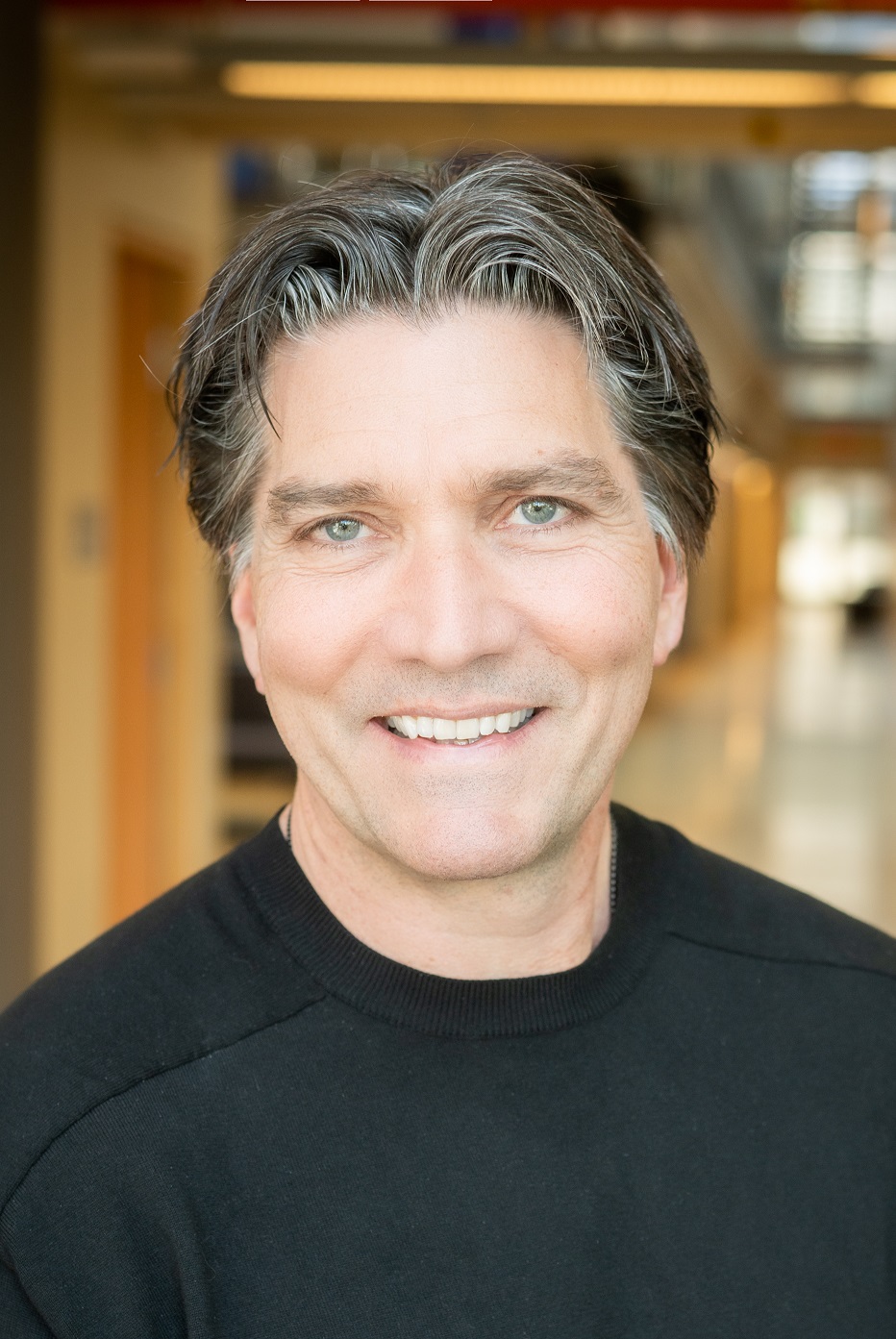}}]{Martin Guay} received a Ph.D. from Queen’s University, Kingston, ON, Canada in 1996. He is currently a Professor in the Department of Chemical Engineering at Queen’s University. His current research interests include nonlinear control systems, especially extremum-seeking control, nonlinear model predictive control, adaptive estimation and control, and geometric control. 
He was a recipient of the Syncrude Innovation Award, the D. G. Fisher from the Canadian Society of Chemical Engineers, and the Premier Research Excellence Award. He is a Senior Editor of IEEE Transactions on Automatic Control. He is the Editor-in-Chief of the Journal of Process Control. He is also an Associate Editor for Automatica and the Canadian Journal of Chemical Engineering. 
\end{IEEEbiography}
\vspace{-15pt}
\begin{IEEEbiography}[{\includegraphics[height=1.25in, clip,keepaspectratio]{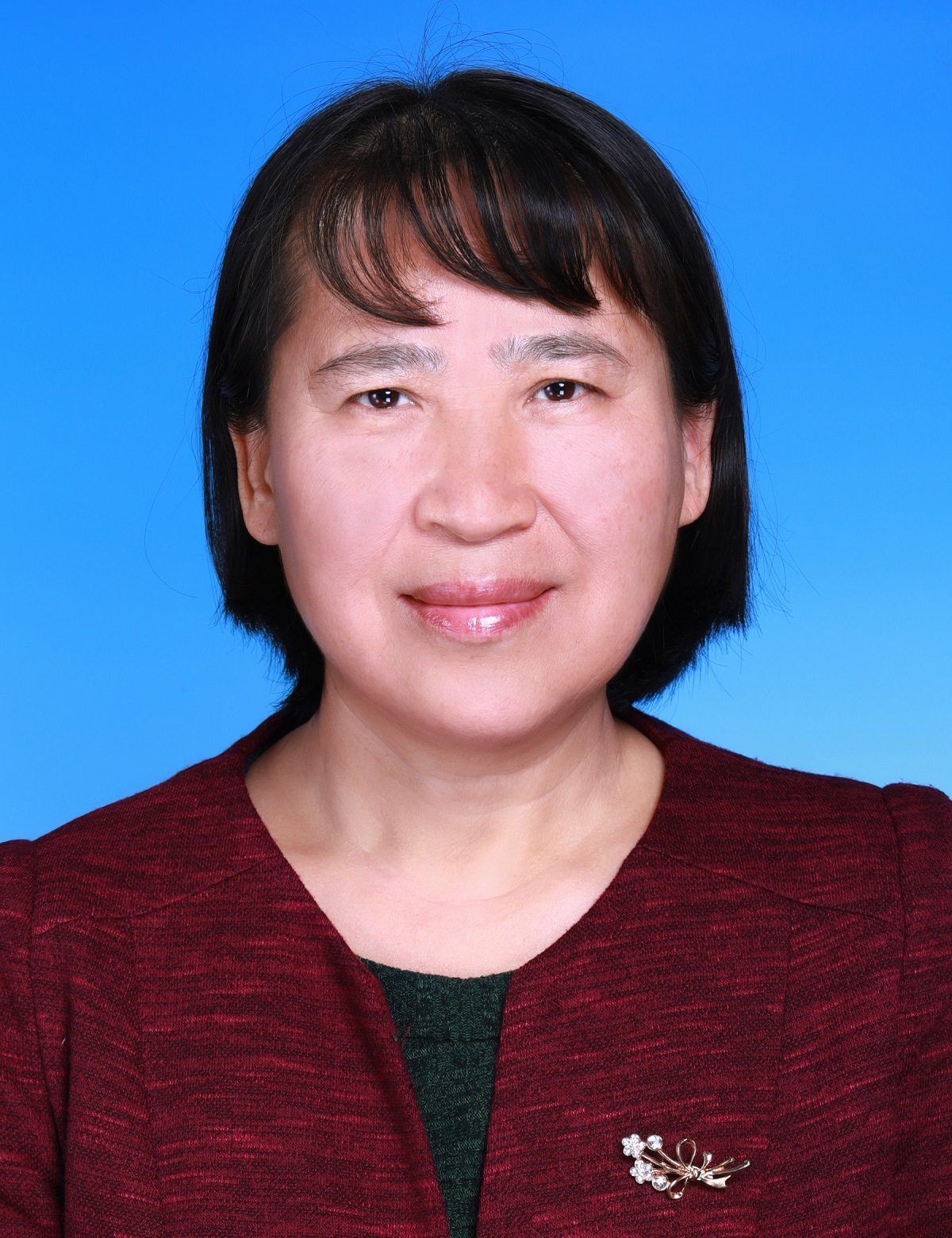}}]{Jinzhi Wang} received the M.S. degree in mathematics from Northeast Normal University, China, in 1988, and the Ph.D. degree in control theory from Peking University, China, in 1998. From 1998 to 2000, she was a Postdoctoral Fellow with the Institute of Systems Science, Chinese Academy of Sciences. 
She is currently a Professor at the Department of Mechanics and Engineering Science, College of Engineering, Peking University. Her research interests include cooperative control of multi-agent systems and control of nonlinear dynamical systems.
\end{IEEEbiography}
\vspace{-15pt}
\begin{IEEEbiography}[{\includegraphics[height=1.25in,keepaspectratio]{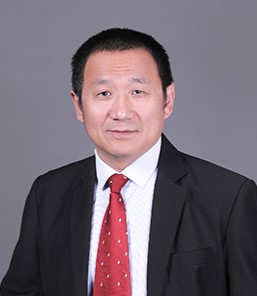}}]{Zhisheng Duan} received the M.S. degree in Mathematics from Inner Mongolia University, China and the Ph.D.degree in Control Theory from Peking University, China in 1997 and 2000, respectively.
From 2000 to 2002, he was a Postdoctoral Fellow at Peking University, where he has been a Full Professor with the Department of Mechanics and Engineering Science, College of Engineering since 2008. He obtained the Outstanding Young Scholar from the National Natural Science Foundation in China and is currently a Cheung Kong Scholar at Peking University. His current research interests include robust control, stability of interconnected systems, nonlinear control, and analysis and control of complex dynamical networks.
Prof. Duan was a recipient of the Guan-Zhao Zhi Best Paper Award at the 2001 Chinese Control Conference and the 2011 First Class Award in Natural Science from the Chinese Ministry of Education. Prof. Duan has been listed by Thomson Reuters and Clarivate Web of Science as Highly Cited Researchers in Engineering since 2017.
\end{IEEEbiography}
\vspace{-15pt}
\begin{IEEEbiography}[{\includegraphics[height=1.25in,keepaspectratio]{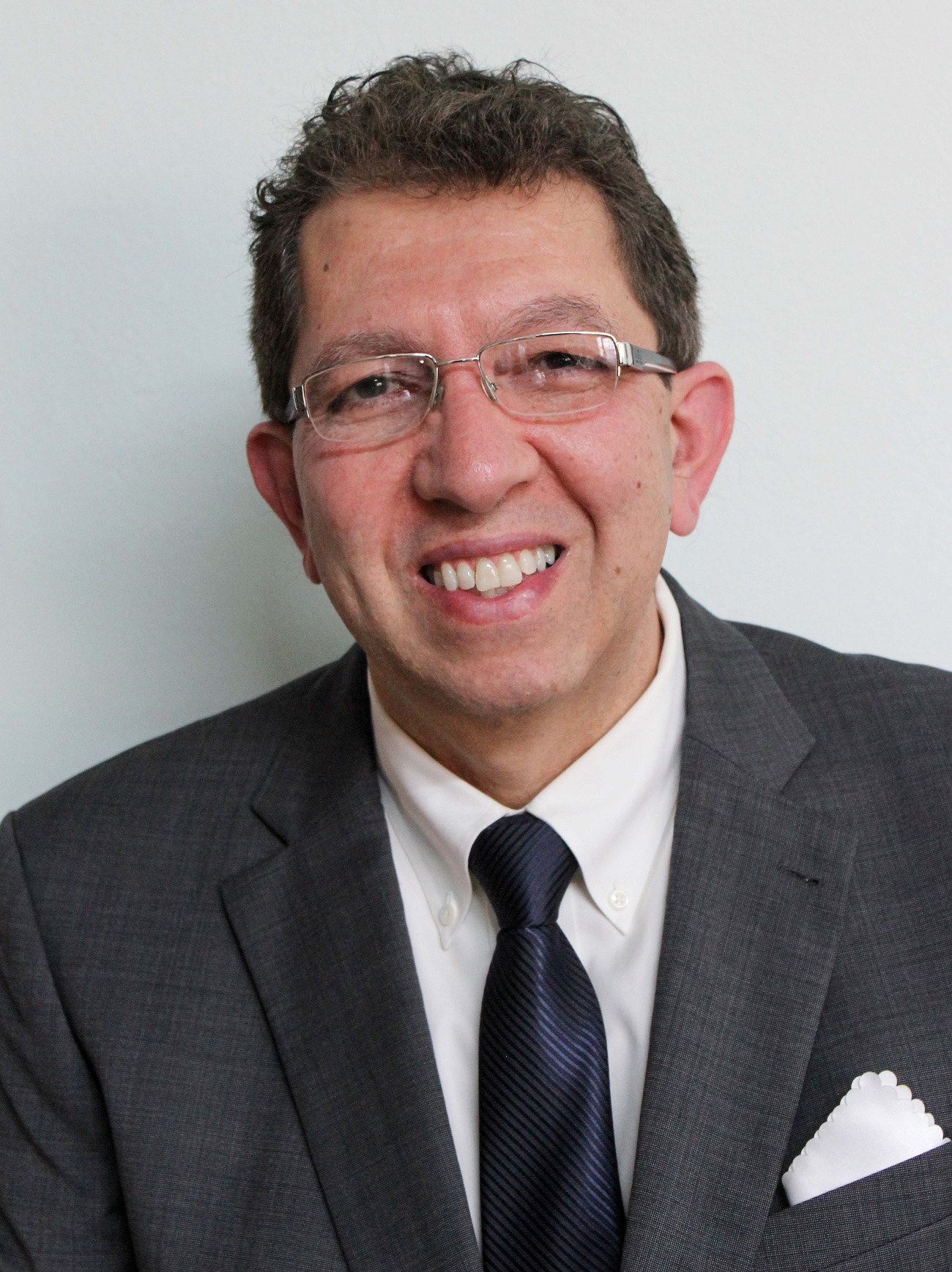}}]{Marios M. Polycarpou} (Fellow of IEEE and IFAC) received the B.A. degree in computer science and the B.Sc. degree in electrical engineering, both from Rice University, USA, in 1987, and the M.S. and Ph.D. degrees in Electrical Engineering from the University of Southern California, USA, in 1989 and 1992 respectively.
He is currently a Professor of Electrical and Computer Engineering and the Director of the KIOS Research and Innovation Center of Excellence at the University of Cyprus. He is also a Founding Member of the Cyprus Academy of Sciences, Letters, and Arts, an Honorary Professor of Imperial College London, and a Member of Academia Europaea (The Academy of Europe). His teaching and research interests are in intelligent systems and networks, adaptive and learning control systems, fault diagnosis, machine learning, and critical infrastructure systems. 
Prof. Polycarpou is the recipient of the 2023 IEEE Frank Rosenblatt Technical Field Award and the 2016 IEEE Neural Networks Pioneer Award. He is a Fellow of the International Federation of Automatic Control (IFAC). He served as the President of the IEEE Computational Intelligence Society (2012-2013), as the President of the European Control Association (2017-2019), and as the Editor-in-Chief of the IEEE Transactions on Neural Networks and Learning Systems (2004-2010). Prof. Polycarpou currently serves on the Editorial Boards of the Proceedings of the IEEE and the Annual Reviews in Control. His research work has been funded by several agencies and industries in Europe and the United States, including the prestigious European Research Council (ERC) Advanced Grant, the ERC Synergy Grant and the EU-Widening Teaming Program.
\end{IEEEbiography}

\end{document}